\documentclass[11pt]{article}
\usepackage[margin=1in]{geometry}

\usepackage[numbers,sort]{natbib}

\usepackage{preamble}
\usepackage[dvipsnames]{xcolor}

\newcommand{\brown}[1]{\textcolor{black}{#1}}

\usepackage[normalem]{ulem}

\newcommand{\E}{\mathbb{E}}

\newcommand{\Rplus}{\mathbb{R}_+}

\newcommand{\Prob}{\mathbb{P}}
\newcommand{\iid}{\overset{\text{i.i.d.}}{\sim}}
\newcommand{\indicator}[1]{\mathbbm{1}\!\!\left\{#1\right\}}
\newcommand{\equalSpace}{\quad\,\,}
\newcommand{\dd}{\mathrm{d}}

\newcommand{\Expo}{\mathrm{Expo}}
\newcommand{\Rev}{\mathrm{Rev}}

\usepackage[most]{tcolorbox}
\usepackage{xcolor}

\definecolor{darkgrey}{RGB}{64, 64, 64}
\definecolor{lightgrey}{RGB}{220, 220, 220}
\newcounter{theorempart}[theorem]

\Crefname{theorempart}{Theorem}{Theorems}

\title{Personalised versus Posted Pricing from Samples%
    \thanks{Accepted at the 22nd Conference on Web and Internet Economics (WINE 2026).}}
\author{Pieter Kleer, Johan van Leeuwaarden and Daan Noordenbos\\\
Department of Econometrics and Operations Research, Tilburg University\\
\texttt{\{p.s.kleer,j.s.h.vanleeuwaarden,d.noordenbos\}@tilburguniversity.edu}}
\date{\today}

\begin{document}
\maketitle
\thispagestyle{empty}
\begin{abstract}
Personalised pricing maximises expected revenue from a market but requires detailed information about individual customers. How much of this revenue can be recovered using a simple posted price based on a finite number of samples from the underlying value distribution? We answer this question by maximising the worst-case ratio between the expected revenues of posted and personalised pricing over the fundamental class of $\lambda$-regular value distributions.

Our results reveal a structural transition as a function of $\lambda$. For the class of monotone hazard rate (MHR) distributions, corresponding to $\lambda = 0$, the sample mean is an optimal statistic: the entire sample can be compressed into its average without any loss of revenue. Beyond the MHR class, corresponding to $\lambda > 0$, this property disappears. We show that the sample mean is no longer optimal, revealing that optimal sample-based pricing rules become substantially more intricate. Nevertheless, we show that a remarkably simple order-statistic based pricing rule is asymptotically optimal as the number of samples $n$ grows, achieving the optimal approximation ratio up to a tight error of order $1/n$.

Our analysis combines techniques from probability, approximation theory and optimization, including doubly infinite linear programming, hypergeometric functions, and combinatorial identities involving incomplete Beta functions.
\end{abstract}

\newpage
\pagenumbering{arabic}
\setcounter{page}{1}

\section{Introduction}
Digital platforms increasingly collect detailed information about individual consumers. In principle, such information enables \emph{personalised pricing}, where each customer is charged a price tailored to their willingness to pay. Under perfect information, personalised pricing extracts the entire surplus from a market and therefore achieves the maximum possible revenue. Examples include online retail, travel platforms, digital advertising, and subscription services, where firms routinely use customer-specific information to guide pricing decisions. 

Despite its revenue advantages, personalised pricing remains controversial \cite{seele2021mapping}. Consumers often perceive differential pricing as unfair and regulators increasingly scrutinise its use. Moreover, the information required to implement it may be costly or unavailable. Consequently, many firms continue to rely on a much simpler and more transparent alternative: a single price offered to all customers. This naturally raises the following question: How much of the revenue attainable under perfect personalised pricing can be recovered using a simple posted price?

To answer this question, one must first specify what information is available to the seller.
Classical revenue maximization problems assume complete knowledge of the underlying value distribution of the customers, but this is often deemed unrealistic in practice. 
In this work we take a sample-based perspective, by assuming to have been given $n$ samples from the underlying value distribution. 
The goal is to design a (possibly randomised) pricing rule that uses a given number of samples as input and whose expected revenue is within a multiplicative factor of the revenue obtained by a personalised pricing that charges each customer their maximum willingness to pay.
In the latter, the expected revenue equals the mean of the underlying value distribution.
This sample-based information framework is standard in mechanism design and revenue maximization as a way to model limited information about the distribution. We refer to the related work section for further discussion and other limited information frameworks such as the use of distributionally robust optimization \cite{elmachtoub2021value} for personalised versus posted pricing.

% \jvl{And we are still lacking some exciting motivation for why pushing things beyond basic MHR to regular is a big deal...} \pk{Put a sentence at the bottom trying to address this; see above in green for quotes form other papers.}

% \jvl{And why the reponse to more or increasingle many samples is such a big deal, and that things changes dramatically when we go beyond one sample! Other policy, math challenges, yet relevant!}
% \pk{Will try to emphasise this more in "Contributions"}

Without further assumptions on the valuation distribution, no non-trivial guarantee is possible, i.e., the performance of a posted pricing rule can be arbitrarily bad compared to a personalised pricing rule. 
To obtain meaningful guarantees, we therefore impose structure on the underlying distribution by studying the {well-known class of $\lambda$-regular valuation distributions \cite{mares2011near,ewerhart2013regular,cole2014sample,schweizer2019performance}}. Following \cite{schweizer2019performance}, we say that a distribution with probability density function $f(x)$ and cumulative density function $F(x)$ is $\lambda$-regular if the generalised hazard rate $r_{\lambda}(x) = f(x)/(1-F(x))^{1+\lambda}$ is non-decreasing.  This parametrised class of distributions interpolates nicely between monotone hazard rate (MHR) distributions ($\lambda = 0$), and regular distributions ($\lambda = 1$), two famous classes of distributions in the field of mechanism design. \brown{The class of $\lambda$-regular distributions makes it possible to go beyond the restrictive aspects of monotone hazard rate distributions, such as their inability to accommodate heavier-than-exponential tails; moreover, non-monotone hazard rate value distributions have been observed in practice, see \cite{ewerhart2013regular} and references therein.}
In Section \ref{sec:related-work}, we review notions from the literature that are equivalent to $\lambda$-regularity and discuss further applications.
% In Section \ref{sec:related-work}, we discuss notions appearing in the literature that are equivalent to $\lambda$-regularity and further use cases.
% In Section \ref{sec:related-work} we discuss notions that are equivalent to $\lambda$-regularity that have appeared in the literature.

%\dn{ik laat dit weg} \pk{It might be nice to have still to make the class more tangible at this point and the fact that we talk at the end of Section 1.2 about equivalent definitions, which remains a bit vague if $\lambda$-regularity has not yet been defined. }

\subsection{Contributions}
% \jvl{Also here we really need to stage the extension beyond MHR, and why this is challenging and relevant. Profoundly changing...}
% \pk{I rewrote things so that first $\lambda$-regular, many sample result is pitched. }

Our first main result concerns the class of MHR distributions.
In this setting we obtain a complete characterization of the problem.
We show in Theorem \ref{thm:mhr-pricing} that posting the sample mean of the observed valuations achieves a fraction $(n/(n+1))^{n+1}$ of the personalised-pricing revenue, and that no deterministic or randomised pricing policy can do better.
Thus, for MHR distributions the sample mean emerges as the optimal way to convert sample information into a posted price. As the number of samples grows, the guarantee converges to $1/e$, which is the best possible guarantee with full information for MHR distributions \cite[Lemma 9]{cole2017applications}.
\brown{The exact result of Theorem \ref{thm:mhr-pricing} captures simultaneously the single- and many-sample regimes, that often require different proof techniques, for example, as shown in \cite{dhangwatnotai2010revenue,babaioff2018two,daskalakis2020more,allouah2022pricing} where instead the optimal monopoly price revenue is used as benchmark. This means we simultaneously get revenue guarantees for only a couple of samples, and can compute how many samples are needed to obtain a given precision for approximating the optimal full information guarantee.} As one of our main technical tools to prove our upper bound result that no pricing rule can do better, we use a doubly infinite linear program representation of the problem and construct a judiciously chosen dual solution yielding the results based on weak duality. 

With a full understanding of the MHR case, we turn to the \brown{less restrictive} and more challenging class of $\lambda$-regular distributions. {In this case the sample mean policy is no longer optimal in the many-sample regime, and the structure of the optimal mechanism appears to be very intricate. This naturally raises the question of whether simple pricing rules can achieve strong performance guarantees.}
Inspired by \cite{allouah2022pricing}, we focus on a natural and practically appealing class of policies based on order statistics.

We first identify for every fixed sample size $n$ the order statistic that maximises the performance ratio among all order statistic policies.
We show that the performance of this optimal order statistic policy is within $O(1/n)$ of the full-information benchmark $(1-\lambda)^{1/\lambda}$ (see e.g., \cite[Lemma 9]{cole2017applications}). Moreover, we show that this convergence rate is optimal among all, possibly randomised, policies.
% We first identify for every fixed sample size $n$ the optimal order statistic, among all order statistics, 
% and show that this policy approaches at an optimal rate of $1/n$, as $n \rightarrow \infty$, the best achievable guarantee of $(1-\lambda)^{1/\lambda}$ (see e.g., \cite[Lemma 9]{cole2017applications}) when the underlying valuation distribution is known.  
These results are summarised in Theorems \ref{thm:order-statistic-policy} and \ref{thm:general-pricing-impossibility}, and show that the best order statistic pricing rule is asymptotically optimal.
To prove these theorems we rely, among other tools, on results from theoretical statistics, and a collection of combinatorial identities involving {Beta functions \cite[Section 6.2]{abramowitz1948handbook} and incomplete Beta functions \cite[Section 6.6]{abramowitz1948handbook}}

{To gain further insight into how the performance guarantee depends on $\lambda$, and in particular how heavier tails affect performance, we also study the single-sample regime in greater detail for $\lambda\in(0,1)$.}
% \brown{To further understand how the performance guarantee changes as $\lambda$ increases, allowing for distributions with heavier tails, we also study the single sample regime for general $\lambda \in (0,1)$.} 
We show in Theorem \ref{thm:single-sample-det} that the natural policy of using the sample value as the posted price
% , in line with the sample mean policy for $n > 1$ samples, 
is optimal among all \textit{deterministic} pricing policies and yields a guarantee of $(1-\lambda)/(4-2\lambda)$. 
The proof relies on significantly different techniques than for the MHR case, {in particular, we rely on hypergeometric functions and their properties.} 
Determining the optimal randomised policy, a highly non-linear optimization problem with (seemingly) multiple local optima, remains open, but we do provide an upper bound on the performance of randomised policies in Appendix \ref{appendix:numerical-bounds} \brown{that is nearly tight for small values of $\lambda$}. 

Beyond the individual technical results and the diversity of proof techniques, we believe the paper contributes a new perspective on pricing with limited information. Whereas some pricing-from-samples literature, including the influential work of Allouah et al.~\cite{allouah2022pricing}, evaluates sample-based policies against the optimal monopoly price, we benchmark against personalised pricing itself. This benchmark is economically natural, directly measures the value of information, and leads to a number of sharp guarantees that appear difficult to obtain under the traditional monopoly-revenue benchmark. {For example, when comparing to the optimal monopoly price, it is to the best of our knowledge, \brown{despite various efforts \cite{dhangwatnotai2010revenue,babaioff2018two,daskalakis2020more,allouah2022pricing},} not known what the best possible (deterministic or randomised) guarantee is for MHR distributions, nor what the best possible deterministic guarantee  is for one available sample and arbitrary $\lambda \in (0,1)$. In contrast, our benchmark allows us to derive tight optimality guarantees in these settings and asymptotically tight guarantees throughout the broader class of $\lambda$-regular distributions with many samples.}
\subsection{Related work}
\label{sec:related-work}

Our work contributes to the literature on robust and data-driven pricing. 
Robust pricing problems, where the seller has partial knowledge of the value distribution and seeks guarantees against worst-case distributions, have been studied extensively. 
Early contributions include \cite{bergemann2008pricing,eren2010monopoly}, which characterise optimal pricing policies under minimal distributional information. Subsequent work has considered richer ambiguity structures, including moment-based information and non-parametric families, see for example \cite{azarmicali2013, azar2013optimal, kos2015selling, carrasco2018optimal, elmachtoub2021value, chen2022distribution, giannakopoulos2023robust, allouah2023optimal, van2026distributionally}.
These works typically assume that uncertainty is described beforehand through a known ambiguity set and design pricing rules that are robust with respect to all distributions in this set. 
Most relevant for us is the work of Elmachtoub et al. \cite{elmachtoub2021value} who study the revenue guarantee of posted pricing rules compared to that of personalised pricing for various ambiguity sets based on simple summary statistics. In contrast, we study a data-driven setting in which the seller observes i.i.d.\ samples from an unknown $\lambda$-regular valuation distribution and evaluates performance relative to the mean valuation of that distribution.

While sample and moment-based information are not directly comparable, sample-based information generally provides a less conservative and more average-case perspective.
For instance, in the classical prophet inequality problem, \citet{rubinstein2020optimal} show that a single sample per distribution suffices to match the $1/2$ guarantee of the full-information benchmark, whereas \citet{correa2026informativeness} demonstrate that moment-based information cannot improve beyond a $1/\log n$ guarantee in general. This highlights the fundamentally different informational nature of sample-based approaches.

Our paper is most closely related to the literature on sample-based pricing. The seminal work of \citet{dhangwatnotai2010revenue} showed that posting a single observed sample as the price achieves a constant approximation to optimal revenue for regular distributions. 
This work was subsequently extended along several lines: to MHR distributions and many samples setting by \citet{huang2018making}, and to the challenging two sample setting by Babaioff et al.~\cite{babaioff2018two} and Daskalakis and Zampetakis \cite{daskalakis2020more}.
Of particular relevance is \citet{allouah2022pricing}, who study sample-based pricing under $\lambda$-regularity and derive approximation guarantees relative to the optimal (monopoly) revenue for any fixed number of samples.
In contrast, we evaluate pricing rules relative to the mean valuation. This benchmark enables an explicit characterization of the optimal guarantee in the MHR case and asymptotically tight bounds for general $\lambda$-regular distributions. Unlike \citet{allouah2022pricing}, whose guarantees are primarily characterised numerically, our approach yields explicit analytical characterizations.

Since our benchmark is the mean valuation rather than monopoly revenue, our work is also related to the broader literature on revenue-to-welfare guarantees in mechanism design. \citet{kleinberg2013ratio} derive such guarantees under $c$-boundedness and hyper-regularity assumptions in Bayesian single-parameter settings. As later observed by \citet{schweizer2019performance}, these assumptions can be naturally unified through the notion of $\lambda$-regularity ($\lambda<1$). 

Finally, beyond pricing-from-samples, $\lambda$-regularity has become a standard structural assumption in mechanism design. Introduced in \citet{schweizer2019performance}, it is equivalent to $\rho$-concavity of \citet{mares2011near} and \citet{ewerhart2013regular} with $\rho=-\lambda$, and to $\alpha$-strong regularity of \citet{cole2014sample} with $\alpha=1-\lambda$. It has been used in a variety of settings beyond pricing, including auction design and mechanism design with limited information \cite{cole2017applications,mares2014analysis,giannakopoulos2021optimal,ge2025optimal,allouah2023optimal}.

\section{Model}
\label{sec:model}

The following notation is fixed throughout: $\mathbb{R}_+$ denotes the non-negative real numbers, $\mathcal{P}(A)$ denotes the set of all probability measures over a set $A$, and a bold variable, like $\boldsymbol{x}$, denotes a vector $(x_1, \dots, x_n)$.

We consider a seller who wishes to sell a single indivisible good to a single buyer. The buyer’s private value $V$ is drawn from an unknown distribution represented by a differentiable cumulative distribution function $F$  supported on $\Rplus$. 
The seller does not know $F$, but has access to $n$ i.i.d. samples $(V_1,\dots,V_n)$ from $F$. The problem we study is how the seller can use the given samples (which are realizations of $V_1,\dots,V_n$) to design a pricing rule whose expected revenue constitutes a large fraction of the buyer's mean valuation $\E[V]$. %We adopt a max-min formulation. 

A pricing rule is a function\footnote{Formally, the seller chooses a Markov kernel $\kappa:\Rplus^n \to \mathcal{P}(\Rplus)$ from the set of all such Markov kernels $\mathcal{K}(\Rplus^n,\Rplus)$. The use of Markov kernels is necessary to avoid measurability issues. A more detailed description of Markov kernels is in Appendix \ref{appendix:eks-proof}.} $\kappa:\Rplus^n \to \mathcal{P}(\Rplus)$ that maps every vector $\boldsymbol{v} = (v_1,\dots,v_n) \in \Rplus^n$ of realised samples to a (possibly) randomised pricing rule, which is a distribution over the non-negative reals. We write $Y_{v_1,\dots,v_n}  = Y_{\boldsymbol{v}} \sim\kappa(\boldsymbol{v})$ for a random variable denoting this price. Note that $Y_{\boldsymbol{v}}$ is deterministic if the pricing rule is deterministic.

Given a value distribution $F$ and posted price $p \in \Rplus$, the expected revenue from selling the item whenever $V \geq p$, with $V \sim F$, is
$
\Rev(F,p):=p(1 - F(p)).
$
The expected revenue of a randomised pricing rule $Y_{\boldsymbol{v}} \sim \kappa(\boldsymbol{v})$ is then $\E[\Rev(F,Y_{V_1,\dots,V_n})]$ and we define its performance relative to the mean of the distribution as
$$
T(\kappa,F):=\frac{\E[\Rev(F,Y_{V_1,\dots,V_n})]}{\E[V]},
$$
where $V,V_1\dots,V_n\iid F$. Note that the expectation in the nominator is taken with respect to both the samples $V_1,\dots,V_n$ as well as the possible randomisation of the pricing rule $Y_{\boldsymbol{v}}$. 

We want to solve the maximin problem 
$
\sup_{\kappa\in \mathcal{K}(\Rplus^n,\Rplus)}\inf_{F\in\mathcal{P}} T(\kappa,F)$,
where first the seller chooses a sample-based pricing rule $\kappa$, after which an adversarial nature  chooses a distribution $F \in \mathcal{P}$ that minimises the ratio $T(\kappa,F)$ with $\mathcal{P}$ a set of admissible probability distributions. Without any restrictions on the set $\mathcal{P}$, no non-zero guarantee is possible for the maximin ratio, even for a seller with complete knowledge of the distribution.\footnote{Consider the distribution with $F(x)=1-(1+x)^{-\alpha}$ for $\alpha > 1$. The mean is $1/(\alpha-1)$ and the optimal posted price is $1/(\alpha-1)$ too. Taking $\alpha\downarrow 1$ shows that no non-zero performance guarantee is possible in general.}

Following \cite{allouah2022pricing}, we focus on the set $\mathcal{P} = \mathcal{P}_{\lambda}$ of $\lambda$-regular distributions for a given $\lambda \in [0,1]$. This set interpolates between distributions with a monotone hazard rate for $\lambda=0$, and regular distributions for $\lambda=1$.

%Following \cite{allouah2022pricing}, we focus on the class of $\lambda$-regular distributions, which interpolates between several standard distributional assumptions studied in the revenue maximization literature.
\begin{definition}[$\lambda$-regularity \cite{schweizer2019performance}]
	Let $F$ be a distribution with density $f$. For $\lambda\in[0,1]$, define the generalised hazard rate $r_\lambda(x)=f(x)/(1-F(x))^{1+\lambda}$. We say $F$ is \emph{$\lambda$-regular} if $r_\lambda$ is non-decreasing on the support of $F$. The set of all $\lambda$-regular distributions is denoted by $\mathcal{P}_\lambda$.
\end{definition}

Assuming $r_{\lambda}$ is a given function, we can solve the generalised hazard rate equation for $F(x)$ resulting in the following lemma.

\begin{lemma}[\cite{schweizer2019performance}, Lemma 1]
    A distribution $F$ is $\lambda$-regular if and only if it can be written as
    $$
    F(x) = 1-\left(1+\lambda R_\lambda(x)\right)^{-1/\lambda}
    $$
     where $R_\lambda(x)=\int_0^xr_\lambda(t)\dd t$ is an increasing, convex function with $R_{\lambda}(0) = 0$.
\end{lemma}

The goal of this work will be to  quantify
\begin{align}
Q_{n,\lambda} :=\sup_{\kappa\in \mathcal{K}(\Rplus^n,\Rplus)}\inf_{F\in\mathcal{P}_{\lambda}} \; \frac{\E[\Rev(F,Y_{V_1,\dots,V_n})]}{\E[V]}.
    \label{ref:main_ratio_lambda}
\end{align}
and to identify the optimal sample-based pricing rules.

\section{Monotone hazard rate}
\label{sec:mhr}
We first study the case $\lambda=0$ corresponding to monotone hazard rate (MHR) distributions.
For this case, we obtain an exact characterisation of the performance ratio $Q_{n,0}$, which is stated in Theorem \ref{thm:mhr-pricing}.

\begin{theorem}
\label{thm:mhr-pricing}
	For $n\ge 1$, 
$$
Q_{n,0}:=\sup_{\kappa\in \mathcal{K}(\Rplus^n,\Rplus)}\inf_{F\in\mathcal{P}_0}\frac{\E[\Rev(F,Y_{V_1,\dots,V_n})]}{\E[V]}=\left(\frac{n}{n+1}\right)^{n+1},
$$
where $V,V_1\dots,V_n\iid F$, $Y_{v_1,\dots,v_n}\sim\kappa((v_1,\dots,v_n))$. The supremum is attained by the pricing rule that is the sample mean, i.e., $Y_{V_1,\dots,V_n} = \frac{1}{n}\sum_{i=1}^n V_i$.
	\refstepcounter{theorempart}\label{thm:mhr-pricing-rule} 
    	\refstepcounter{theorempart}\label{thm:mhr-pricing-impossibility} 
\end{theorem}

The proof proceeds in two parts.
For the lower bound, we analyse the performance of the sample mean pricing rule using convexity arguments and integration by parts.
To obtain the matching upper bound, we start by restricting to 
constant hazard rate distributions, which are exponential distributions of the form $F_\theta(x)=1-e^{-\theta x}$ for $\theta > 0$. Restricting to this class yields an upper bound on $Q_{n,0}$ described by the (doubly) infinite-dimensional linear program
    \begin{align*}
        &\sup && C\\
        &\text{s.t.} &&\E[\Rev(F,Y_{Z_1\dots,Z_n})]\ge C \cdot \E[Z] \; \text{ for all } \theta>0,
        \\
        &&& \kappa((z_1,\dots,z_n))\in\mathcal{P}(\mathbb{R}_+)\text{ for all }(z_1,\dots,z_n)\in\mathbb{R}_+^n\text{, } C\in\mathbb{R}.
    \end{align*}
\noindent The dual program is given by
    \begin{align*}
        &\inf && \int_{\mathbb{R}_+^n}\lambda(\boldsymbol{z})\dd \boldsymbol{z}\\
        &\text{s.t.} &&
        \lambda(\boldsymbol{z})\ge\int_0^\infty \theta^{n+1}s\exp\left(-\theta\left(s+\sum_{i=1}^nz_i\right)\right)\dd \mu(\theta)\text{ for all }\boldsymbol{z}\in\mathbb{R}_+^n\text{ and }s\ge 0,
        \\
        &&& \mu\in\mathcal{P}((0,\infty))\text{, } \lambda:\mathbb{R}_+^n\to\mathbb{R}.
    \end{align*} 
Weak duality implies that any dual feasible solution is an upper bound on the primal objective function value $C$. With a judiciously chosen sequence of feasible dual solutions we can establish an upper bound of $\left(n/(n+1)\right)^{n+1}$ on $Q_{n,0}$ in this way.
The details of this are in Appendix \ref{appendix:eks-proof}. 

We conclude this section with the analysis of the sample mean policy.

\begin{proof}[Proof of Theorem \ref{thm:mhr-pricing} (Analysis of sample mean pricing rule)]

Recall that $F$ is $0$-regular, i.e., has a monotone hazard rate, if $r_0(x)=f(x)/(1-F(x))$ is non-decreasing.
    Consider the pricing rule which sets the price equal to the sample mean. The expected revenue for this pricing rule is
    \begin{align*}
        \mathrm{Rev}_{\text{SM}}(F)&:=\E\left[\mathrm{Rev}\left(F,\frac{1}{n}\sum_{i=1}^n V_i\right)\right]\\
        &=\int_{\Rplus^n}\left(\frac{1}{n}\sum_{i=1}^nv_i\right)S\left(\frac{1}{n}\sum_{i=1}^nv_i\right)\left(\prod_{i=1}^nf(v_i)\right)\dd v_1\dots\dd v_n,
    \end{align*}
    where $S(x)=1-F(x)$ is the survival function. Let $H(x)=\int_0^xr_0(t)\dd t$. Because $-S'(t)/S(t)=r_0(t)$ we have that $S(x)=\exp(-H(x))$. Moreover, observe that $H$ is convex since $r_0$ is non-decreasing. By this convexity,
    $$
        H\left(\frac{1}{n}\sum_{i=1}^nv_i\right)\le \frac{1}{n}\sum_{i=1}^nH\left(v_i\right)\implies S\left(\frac{1}{n}\sum_{i=1}^nv_i\right)\ge \prod_{i=1}^nS(v_i)^{1/n},
    $$
    because $S$ is a decreasing function in $x$.
    Consequently,
    \begin{align*}
        \mathrm{Rev}_{\text{SM}}(F)&\ge \int_{\Rplus^n}\left(\frac{1}{n}\sum_{i=1}^nv_i\right)\prod_{i=1}^nS(v_i)^{1/n}\left(\prod_{i=1}^nh(v_i)S(v_i)\right)\dd v_1\dots\dd v_n\\
    &=\int_{\Rplus^n}\left(\frac{1}{n}\sum_{i=1}^nv_i\right)\left(\prod_{i=1}^nh(v_i)S(v_i)^{1+1/n}\right)\dd v_1\dots\dd v_n\\
    &=\frac{1}{n}\sum_{i=1}^n\int_{\Rplus^n}v_i\left(\prod_{i=1}^nh(v_i)S(v_i)^{1+1/n}\right)\dd v_1\dots\dd v_n\\
    &=\int_{\Rplus^n} v_1\left(\prod_{i=1}^nh(v_i)S(v_i)^{1+1/n}\right)\dd v_1\dots\dd v_n\\
    &=\int_0^\infty v h(v)S(v)^{1+1/n}\dd v\left(\int_0^\infty h(v)S(v)^{1+1/n}\dd v\right)^{n-1}.
    \end{align*}
    Using $f(v_i)=h(v_i)S(v_i)$ in the first step. The penultimate step follows from symmetry. Recall that $S'(v)=-h(v)S(v)$, so with the substitution $t=S(v)$ we obtain
    \begin{align*}
        \int_0^\infty h(v)S(v)^{1+1/n}\dd v&=-\int_0^\infty S(v)^{\frac{1}{n}}S'(v)\dd v=-\int_1^0t^{\frac{1}{n}}\dd t=\frac{n}{n+1}.
    \end{align*}
    The antiderivative of $S'(v)S(v)^{1/n}$ is $\frac{n}{n+1}S(v)^{1+1/n}$, so by integration by parts,
    \begin{align*}
        \int_0^\infty vh(v)S(v)^{1+1/n}\dd v&=-\int_0^\infty v S'(v)S^{\frac{1}{n}}\dd v\\
        &= \int_0^\infty \frac{n}{n+1}S(v)^{1+1/n}\dd v-\left[v\frac{n}{n+1}S(v)^{1+1/n}\right]_0^\infty\\
        &=  \frac{n}{n+1} \int_0^\infty S(v)^{1+1/n}\dd v\\
        &\ge \frac{n}{n+1}\int_0^\infty S\left(\frac{n+1}{n}v\right)\dd v&\text{(MHR )} \\
        &=\left(\frac{n}{n+1}\right)^2\int_0^\infty S(w)\dd w=\left(\frac{n}{n+1}\right)^2\E[V].
    \end{align*}
    Combining the expressions yields
    $$
        \int_0^\infty v h(v)S(v)^{1+1/n}\dd v\left(\int_0^\infty h(v)S(v)^{1+1/n}\dd v\right)^{n-1}\ge \left(\frac{n}{n+1}\right)^{n+1}\E[V]
    $$
    The above holds for all $F$, so the sup-inf problem is at least $\left(n/(n+1)\right)^{n+1}$.
    
\end{proof}

\section{Single sample and \texorpdfstring{$\lambda$}{lambda}-regular}
\label{sec:lambda-single}
For $\lambda>0$, analysing $Q_{n,\lambda}$ becomes substantially more difficult, and we are no longer able to give a closed-form expression for arbitrary combinations of $n$ and $\lambda$. In particular, the family of distributions with constant hazard rate, the exponential family that we exploit in Theorem \ref{thm:mhr-pricing}, turns into the more challenging family of Pareto distributions for $\lambda > 0$ that does not have properties as nice as the exponential family for our purposes.
%One reason for this difficulty is that the family of distributions with a constant generalised hazard rate loses much of the tractable structure present in the MHR case. 
%When $\lambda=0$, this family coincides with the exponential distribution, whereas for $\lambda>0$ it becomes the Pareto distribution to which several arguments available in the monotone hazard rate setting no longer apply.
In light of these difficulties, we instead analyse two regimes for $\lambda > 0$: the single sample setting and the many samples setting, where the latter is treated in Section \ref{sec:lambda-multiple}. 

In the single-sample setting considered here, we show that the pricing rule that simply posts the sample $V_1$ as the price guarantees an approximation ratio of $(1-\lambda)/(4-2\lambda)$, and that this is the best deterministic pricing rule.
\begin{theorem} 
    \label{thm:single-sample-det}
	Let $\lambda\in(0,1)$, then
	$$
	   Q^{\mathrm{det}}_{1,\lambda}:=\sup_{\Psi:\mathbb{R}\to\mathbb{R}}\inf_{F\in\mathcal{P}_\lambda}\frac{\E[\Rev(F,\Psi(V_1))]}{\E[V]}=\frac{1-\lambda}{4-2\lambda},
    $$
    where $\Psi$ is a deterministic pricing rule mapping the observed sample to a posted price.
    The supremum is attained by the pricing rule $\Psi(x)=x$.
\end{theorem}
To prove Theorem \ref{thm:single-sample-det} we first argue, essentially using similar reasoning as in \cite[Proposition 1]{allouah2022pricing}, that we can restrict to scaling policies of the form $\Psi(x)=\gamma x$ for $\gamma \geq 0$. 
We then separate the analysis in three parts: $\gamma < 1, \gamma = 1$ and $\gamma > 1$. The case $\gamma = 1$ gives the bound $(1-\lambda)/(4-2\lambda)$ and for $\gamma > 1$, we argue that the resulting inf-problem vanishes. 
Finally, the case $\gamma < 1$ we tackle by showing that no such $\gamma$ can give a ratio higher than $(1-\lambda)/(4-2\lambda)$ using results and properties of hypergeometric functions. This is the technically most challenging part.

Numerical experiments seem to suggest that randomised pricing rules are able to beat the best deterministic pricing rule, but proving this with similar techniques appears intractable at the moment. We leave this as an interesting open question.
As a corollary of Theorem \ref{thm:general-pricing-impossibility}, we do have that randomization can never reach the performance achievable with complete distributional knowledge. 
As Theorem \ref{thm:general-pricing-impossibility} holds for general $n \geq 1$, we remark that the bound for $n=1$ is quite poor.
In Appendix \ref{appendix:numerical-bounds} we establish better bounds which show that the possible improvement from randomisation is quite modest.

In the remainder of this section, we prove Theorem \ref{thm:single-sample-det}.
\begin{proof}[Proof of Theorem \ref{thm:single-sample-det}]
We first establish in the following lemma that we can restrict to pricing rules that scale the sample. 

\begin{lemma}
    \label{lemma:scaling-reduction}
    Under the (mild) regularity assumption that $\lim_{t\to 0}\Psi(t)/t$ exists,\footnote{This is the same regularity assumption as in \cite[Proposition 1]{allouah2022pricing}.} it holds for $\lambda \in (0,1)$ that
    $$
    	Q^{\mathrm{det}}_{1,\lambda}:=\sup_{\Psi:\mathbb{R}\to\mathbb{R}}\inf_{F\in\mathcal{P}_\lambda}\frac{\E[\Rev(F,\Psi(V_1))]}{\E[V]}=\sup_{\gamma > 0}\inf_{F\in\mathcal{P}_\lambda}\frac{\E[\Rev(F,\gamma V_1)]}{\E[V]}.
    $$
\end{lemma}
\begin{proof}
See Appendix \ref{appendix:scaling-proof}.
\end{proof}

\noindent Based on Lemma \ref{lemma:scaling-reduction}, we next split into three cases: $\gamma=1$, $\gamma>1$ and $\gamma <1$. 

    \noindent\textit{Case 1: $\gamma = 1$.} Fix $F\in\mathcal{P}_\lambda$, then
    $$
        \frac{\E[\mathrm{Rev}(F, \gamma V_1)]}{\E[V]}=\frac{\E[V_1(1-F(V_1))]}{\E[V]}=\frac{1}{2}\frac{\E[\min\{V,V_1\}]}{\E[V]}\ge\frac{1-\lambda}{4-2\lambda},
    $$
    where the last inequality follows from \cite[Proposition 10]{schweizer2019performance}.
    Next, take the distribution $G$ with $R_\lambda(x)=x/\lambda$, then
    \begin{align*}
        &\E[V]=\int_0^\infty(1-G(v))\dd v=\int_0^\infty(1+v)^{-1/\lambda}\dd v=\left[\frac{1}{1-1/\lambda}(1+v)^{1-1/\lambda}\right]_0^\infty=\frac{\lambda}{1-\lambda}\\
        &\E[\min\{V,V_1\}]=\int_0^\infty(1-G(v))^2\dd v=\frac{\lambda}{2-\lambda}, 
    \end{align*}
    so $\E[\min\{V,V_1\}]/(2\E[V])=(1-\lambda)/(4-2\lambda)$. Therefore
    $$
    \inf_{F\in\mathcal{P}_{\lambda}}\frac{\E[\mathrm{Rev}(F, V_1)]}{\E[V]}=\frac{1-\lambda}{4-2\lambda}.
    $$

    \noindent\textit{Case 2: $\gamma > 1$.} For $\alpha > 0$, let $G_\alpha$ be the distribution with $R_\lambda(x)=\max\{0,\alpha (x-1)\}$.\footnote{The everywhere differentiable choice $R_\lambda(x)=x^\alpha$ also works, but is harder to calculate with.} 
    Then $1-G_\alpha(x)=(1+\lambda\max\{0,\alpha(x-1)\})^{-1/\lambda}$ for $x\ge 0$ and $g(x)=\alpha (1+\lambda \alpha(x-1) )^{-1/\lambda-1}$ for $x\ge 1$ and zero otherwise. 
    The mean is
    $$
    \E[V]=\int_0^\infty(1-G_\alpha(v))\dd v=1+\int_1^\infty(1+\lambda\alpha(v-1))^{-1/\lambda}\dd v=1+\frac{1}{\alpha(1-\lambda)}
    $$
    and the revenue is bounded by 
    \begin{align*}
        \E[\mathrm{Rev}(G_\alpha, \gamma V_1)]&=\int_0^\infty\gamma v(1-G_\alpha(\gamma v))g(v)\dd v=\gamma\int_1^\infty v(1-G_\alpha(\gamma v))g(v)\dd v\\
        &\le \gamma(1-G_\alpha(\gamma ))\int_1^\infty vg(v)\dd v= \gamma(1-G_\alpha(\gamma ))\E[V].
    \end{align*}
    Therefore 
    $$
    \lim_{\alpha\to \infty}\frac{\E[\mathrm{Rev}(G_\alpha, \gamma V_1)]}{\E[V]}\le \lim_{\alpha\to \infty}\gamma(1-G_\alpha(\gamma ))=  \lim_{\alpha\to \infty}\gamma (1+\lambda\alpha(\gamma-1))^{-1/\lambda}=0.
    $$
    Consequently, for $\gamma > 1$ there is no performance guarantee, that is,
    $$
    \sup_{\gamma > 1}\inf_{F\in\mathcal{P}_{\lambda}}\frac{\E[\mathrm{Rev}(F, \gamma V_1)]}{\E[V]}=0.
    $$

    \textit{Case 3: $\gamma < 1$.} 
    Consider the distribution $G$ again with $R_\lambda(x)=x/\lambda$ and let
        \begin{align*}
        I(\gamma)= \frac{\E[\mathrm{Rev}(G, \gamma V_1)]}{\E[V]} &=\frac{1}{\E[V]}\int_0^\infty \gamma v(1+\gamma v)^{-1/\lambda}\frac{1}{\lambda}(1+v)^{-1/\lambda-1}\dd v\\
        &=\frac{1-\lambda}{\lambda^2}\gamma\int_0^\infty v(1+\gamma v)^{-1/\lambda}(1+v)^{-1/\lambda-1}\dd v
    \end{align*}
    If $I(\gamma)$ is an increasing function, then 
    \begin{equation*}
        \sup_{\gamma \in(0,1)}\inf_{F\in\mathcal{P}_{\lambda}}\frac{\E[\mathrm{Rev}(F, \gamma V_1)]}{\E[V]}\le \sup_{\gamma \in(0,1)}\frac{\E[\mathrm{Rev}(G, \gamma V_1)]}{\E[V]}=I(\gamma)<I(1)=\frac{1-\lambda}{4-2\lambda},
    \end{equation*}
    from which it follows, in combination with Case 1 that no $\gamma < 1$ can attain the supremum in $Q_{1,\lambda}^\mathrm{det}$. It therefore remains to prove  that $I(\gamma)$ is increasing on $(0,1)$, which is done in the following lemma.  
    \begin{lemma}
    \label{lemma:special-function-monotonicity}
    For fixed $\lambda\in(0,1)$, the function 
    $$
    I(\gamma) = \frac{1-\lambda}{\lambda^2}\gamma\int_0^\infty v(1+\gamma v)^{-1/\lambda}(1+v)^{-1/\lambda-1}\dd v
    $$ is increasing on $(0,1)$.
\end{lemma}
The proof of Lemma \ref{lemma:special-function-monotonicity} is deferred to Appendix \ref{app:hyper-geometric}. In a nutshell, we can rewrite $I(\gamma)$ as the product of $\gamma$ and a hypergeometric function by a variable substitution and applying an identity due to Euler. We then use a derivative formula for hypergeometric functions together with a so-called contiguous relation, in the spirit of Gauss' contiguous relations for hypergeometric functions, to argue about the derivative of $I(\gamma)$.
This completes the proof of the theorem.
\end{proof}

\section{Multiple samples and \texorpdfstring{$\lambda$}{lambda}-regular}
\label{sec:lambda-multiple}
When multiple samples are available, the seller can employ pricing rules that depend on the samples in more intricate ways.
Unlike the MHR case, however, where the sample mean is optimal, in the regime where $\lambda \in (0,1)$, different or more complicated strategies seem necessary to obtain the optimal performance. To substantiate this claim, we show in Appendix \ref{appendix:non-optimality-sample-mean} that for $\lambda=1/2$ and $n=2$ samples the sample mean policy is strictly suboptimal. We suspect the optimal mechanism for a given $\lambda$ to be very delicate and potentially lacking interpretability.

Because an interpretable and simple strategy is preferable from a managerial perspective, and motivated by \cite{allouah2022pricing} that also consider a class of order statistics-based pricing rules, we will consider the simple family of rules that post a price equal to an order statistic of the observed samples.
Our motivation for considering order statistics is that they provide, like the mean, an estimate of the scale of the distribution. However, unlike the mean, order statistics do this in a more granular way. Moreover, order statistics remain analytically tractable and flexible enough to adapt to different values of $\lambda$, offering a degree of control that the sample mean does not provide.

Besides the motivation sketched in the previous paragraph, we can in fact show that an order statistic-based pricing rule is asymptotically optimal as the number of samples grows large, i.e., as $n \rightarrow \infty$.

Our first result characterises the performance of the optimal order statistic pricing rule. 
\begin{theorem}
    \label{thm:order-statistic-policy}
    Let $V_{r:n}$ denote the $r$-th order statistic, then
    \begin{align*}
        \max_{r\in\{1,\dots,n\}}\inf_{F \in \mathcal{P}_\lambda} \frac{\mathbb{E}\big[\Rev(F,V_{r:n} ))\big]}{\mathbb{E}[V]} = C(n,r^*,\lambda),
    \end{align*}
    where $r^*=\max\{r:C(n,r,\lambda)>C(n,r-1,\lambda)\}$ and
    $$
    C(n,r,\lambda)=\frac{1-\lambda}{\lambda}\left(\prod_{k=n-r+1}^n\frac{k}{k+1-\lambda}-\frac{n-r+1}{n+1}\right).
    $$
\end{theorem}
We also give the rate at which this performance guarantee converges to the full-information performance guarantee $(1- \lambda)^{1/\lambda}$.
\begin{corollary}
    \label{corol:C-n-r-lambda-asymp}
For the regime where $n$ is large, it holds that $$C(n,r^*,\lambda)=(1-\lambda)^{1/\lambda}-\Theta\left(\frac{1}{n}\right),$$
and $\lim_{n \rightarrow \infty} r^*/n = 1-(1-\lambda)^{1/\lambda}$.
\end{corollary}

Note that, in particular, the quantity $C(n,r^*,\lambda)$ provides for every $n$ a lower bound on $Q_{n,\lambda}$. We complement the above result on the asymptotic behaviour of $C(n,r^*,\lambda)$ by providing a matching upper bound of the same asymptotic order that applies to any (potentially randomised) pricing rule.
\begin{theorem}
    \label{thm:general-pricing-impossibility}
    For every $n\ge 1$ and $\lambda\in(0,1)$,
    $$
    Q_{n,\lambda}\le(1-\lambda)^{1/\lambda}-\frac{1}{384\left(1-\lambda\right)}\left(1+\frac{\lambda}{1-\lambda}\frac{5}{4}\right)^{-\frac{1}{\lambda}-2}\frac{1}{n}.
    $$
\end{theorem}
Together, these results show that the simple and interpretable order-statistic-based pricing is an asymptotically optimal pricing rule. In light of \citet[Proposition 4]{schweizer2019performance}, we strongly suspect that the sample mean is also asymptotically optimal. The advantage of our approach is that, in addition to establishing asymptotic optimality, we obtain exact finite-sample guarantees for the order-statistic rule, which appear not readily available from the corresponding analysis of the sample mean.

Numerical results regarding $C(n,r^*,\lambda)$ and the proofs for the asymptotic results are in Appendix \ref{appendix:lambda-many-samples-asymptotic}.
The proof of Theorem \ref{thm:general-pricing-impossibility} is in Appendix \ref{appendix:proof-general-pricing-impossibility}. It follows similar ideas as in \cite{huang2018making} based on Le Cam's two-point method \cite{tsybakov2008lower}.
The rest of this section is devoted to proving Theorem~\ref{thm:order-statistic-policy}.

\subsection{Proof of Theorem \ref{thm:order-statistic-policy}}
The missing proofs in this section are deferred to Appendix \ref{appendix:order-statistic-policy}.
We denote by $V_{r:n}$ the $r$-th order statistic, i.e., $V_{r:n}$ is the $r$-th smallest value from the $V_1,\dots,V_n$. We will first state various lemmas that are used in the proof, and often rely on the function
\begin{align}
        C(n,r,\lambda)&=\frac{1-\lambda}{\lambda}\left(\prod_{k=n-r+1}^n\frac{k}{k+1-\lambda}-\frac{n-r+1}{n+1}\right)\nonumber\\
        &=
        \frac{1-\lambda}{\lambda}r\binom{n}{r}\left(B(n+2-r-\lambda, r)-B(n+2-r, r)\right),
        \label{eq:C-n-r-lambda}
\end{align}
where $B(a,b) =\int_0^\infty t^{1-a}(1-t)^{1-b}\mathrm{d}t$ is the Beta function.
The equivalence of the above two expressions is established in Appendix \ref{appendix:lambda-many-samples-asymptotic}.
The former representation is better for numerical calculations, while the latter is more convenient for proofs.

\begin{lemma}
    \label{lemma:constant-beta-equivalence}
    For any $1 \leq r \leq n$ and $\lambda \in (0,1)$, it holds that
    $$
    B(n+2-r,r) > (1-\lambda)B(n+2-r-\lambda,r) \; \iff \;  C(n,r,\lambda)>C(n,r-1,\lambda).
    $$
\end{lemma}

\begin{lemma}
    \label{lemma:C-unimodality}
    Let $n\ge 2$. The function $q:\{1,\dots,n\}\to \mathbb{R}$ given by $q(r)=C(n,r,\lambda)$ is unimodal, meaning that the function $q(r) - q(r-1)$ changes sign precisely once.
\end{lemma}

\begin{lemma}
    \label{lemma:order-statistic-beta}
    Let $V_1,\dots,V_n$ be non-negative random variables with distribution $F$ and finite mean, then
    $$
    \E[\Rev(F,V_{r:n})]=r\binom{n}{r}\int_0^\infty B(1-F(x); n-r+2,r)\dd x,
    $$
    where $B(x;a,b) = \int_0^x t^{1-a}(1-t)^{1-b}\mathrm{d}t$ is the incomplete Beta function. 
\end{lemma}

\begin{lemma}
    \label{lemma:const-gen-hazard-order-stat}
    Fix $\lambda\in(0,1)$ and let $F(x)=1-(1+x)^{-1/\lambda}$, then for $V_1,\dots V_n\sim F$
    $$
    \frac{\E[\Rev(F,V_{r:n})]}{\E[V]}=C(n,r,\lambda),
    $$
    with $C(n,r,\lambda)$ as in \eqref{eq:C-n-r-lambda}.
    % expressed in terms of the beta function $B(a,b) = B(\infty,a,b) = \int_0^\infty t^{1-a}(1-t)^{1-b}\mathrm{d}t$:
    % $$
    % C(n,r,\lambda)=\frac{1-\lambda}{\lambda}r\binom{n}{r}\left(B(n+2-r-\lambda, r)-B(n+2-r, r)\right).
    % $$
\end{lemma}
We now continue with the proof based on the above lemmas in combination with a technique used in \cite[Proposition 10]{schweizer2019performance} that relies on a lemma from \citet[Chapter 4, Lemma 7.1]{barlow1981statistical}.

%With the above two identities and the technique used by \cite{schweizer2019performance} to prove Proposition 10, we prove the theorem.
\begin{cleverproof}{thm:order-statistic-policy}
    First of all, by Lemma \ref{lemma:C-unimodality} the map $r\mapsto C(n,r,\lambda)$ is unimodal, so it is maximised for $r^*=\max\{r:C(n,r,\lambda)>C(n,r-1,\lambda)\}$. Lemma \ref{lemma:const-gen-hazard-order-stat} gives an upper bound on the inner infimum, so
    $$
    \max_{r\in\{1,\dots,n\}}\inf_{F \in \mathcal{P}_\lambda} \frac{\mathbb{E}\big[\Rev(F,V_{r:n} ))\big]}{\mathbb{E}[V]}\le \max_{r\in\{1,\dots,n\}}C(n,r,\lambda)=C(n,r^*,\lambda).
    $$
    We now derive a matching lower bound, meaning we need to show for all $F \in \mathcal{P}_{\lambda}$ and $r = r^*$ that
    \begin{align}
    \E[V_{r:n}(1-F(V_{r:n}))]-C(n,r,\lambda)\E[V] \geq 0.
        \label{eq:C-lower-bound}
    \end{align}
    In fact, we will show \eqref{eq:C-lower-bound} for every $r \leq r^*$, and because the maximum of $C(n,r,\lambda)$ is attained at $r = r^*$, $C(n,r^*,\lambda)$ it is the best lower bound we can obtain in this way.
    We first define 
     \begin{align*}
    g(x):=A(1-F(x))(1-F(x))^{-\lambda}f(x) \text{ with } A(t)=r\binom{n}{r}\frac{B(t;n-r+2,r)}{t}-C(n,r,\lambda),
     \end{align*}
    and $h(x) := 1/r_{\lambda}(x) = (1-F(x))^{\lambda + 1}/f(x)$ so that, using Lemma \ref{lemma:const-gen-hazard-order-stat} in the first equality,
    \begin{flalign*}
    & \E[V_{r:n}(1-F(V_{r:n}))]-C(n,r,\lambda)\E[V]\\
     & \ \  = r\binom{n}{r}\int_0^\infty B(1-F(x); n-r+2,r)\dd x-C(n,r,\lambda)\int_0^\infty(1-F(x))\dd x\\
      & \ \  =\int_0^\infty A(1-F(x))(1-F(x))\dd x\\
       & \ \   =\int_0^\infty A(1-F(x))\frac{(1-F(x))^{-\lambda}f(x)}{r_\lambda(x)}\dd x\\
        & \ \  = \int_0^\infty h(x)g(x)\dd x.
    \end{flalign*}
To show that $\int_0^\infty h(x)g(x)\mathrm{d}x \geq 0$, which establishes \eqref{eq:C-lower-bound} for a given $r$, we employ the following lemma, whose proof is included in Appendix \ref{appendix:order-statistic-policy} for completeness.
\begin{lemma}[Chapter 4, Lemma 7.1 in \cite{barlow1981statistical}]
    \label{lemma:barlow}
    Let $g$ be a function with $g(0)\ge 0$, $\int_0^\infty g(x)\dd x=0$ and that changes sign at most once in $[0,\infty)$ and let $h$ be a non-increasing function, then
    $$
    \int_0^\infty h(x)g(x)\dd x\ge 0.
    $$
\end{lemma}

\noindent The remainder of the proof is now devoted to showing that our $g$ and $h$ as defined above satisfy the conditions of Lemma \ref{lemma:barlow}. First, note that the function $r_{\lambda}$ is non-decreasing by definition, and so the function $h$ is therefore non-increasing. We next prove that the three properties required for $g$ hold true.

Note that $g(0) = A(1)f(0)$, because $F(0) = 0$, and so $g(0) \geq 0$ if and only if $A(1) \geq 0$ as $f(0) \geq 0$ is always true for a probability density function. Observe that
    \begin{align}
    \label{eq:A-1-condition}
        A(1)  = \; &r\binom{n}{r}B(n-r+2,r)-C(n,r,\lambda)>0\nonumber\\
        \iff \; &r\binom{n}{r}B(n-r+2,r)-\frac{1-\lambda}{\lambda}r\binom{n}{r}\left(B(n+2-r-\lambda, r)-B(n+2-r, r)\right)>0\nonumber\\
        \iff \; &B(n+2-r,r) > (1-\lambda)B(n+2-r-\lambda,r)\nonumber\\
        \iff \; &C(n,r,\lambda) > C(n,r-1,\lambda),
    \end{align}
where the first equivalence follows by plugging in the definition of $C(n,r,\lambda)$, the second by simplification, and the third by Lemma \ref{lemma:constant-beta-equivalence}. By the unimodality of $r \mapsto C(n,r,\lambda)$ established in Lemma \ref{lemma:C-unimodality} it follows \eqref{eq:A-1-condition} is satisfied for all $r \leq r^*$. To conclude, for all $r \leq r^*$ we have $g(0) > 0$.

Next, using the substitution $s = 1 - F(x)$, so that $\mathrm{d}s = -f(x)dx$, we obtain
    \begin{align*}
        \int_0^\infty g(x)\dd x&=\int_0^\infty A(1-F(x))(1-F(x))^{-\lambda}f(x) \dd x\\
        &= - \int_1^0 A(s)s^{-\lambda} \dd s
        = \int_0^1 A(s)s^{-\lambda} \dd s\\
        &= r\binom{n}{r}\int_0^1 B(s,n-r+2,r)s^{-\lambda-1} \dd s+\frac{C(n,r,\lambda)}{1-\lambda}\\
        &= r\binom{n}{r}\int_0^1 \int_0^s t^{n-r+1}(1-t)^{r-1}s^{-\lambda-1}\dd t \dd s+\frac{C(n,r,\lambda)}{1-\lambda}\\
        &= r\binom{n}{r}\int_0^1  t^{n-r+1}(1-t)^{r-1}\int_t^1s^{-\lambda-1}\dd s \dd t+\frac{C(n,r,\lambda)}{1-\lambda}\\
        &= \frac{r}{\lambda}\binom{n}{r}\int_0^1  t^{n-r+1}(1-t)^{r-1}(t^{-\lambda}-1)\dd t+\frac{C(n,r,\lambda)}{1-\lambda}\\
        &= \frac{r}{\lambda}\binom{n}{r}\left(B(n+2-r-\lambda, r)-B(n+2-r, r)\right)+\frac{C(n,r,\lambda)}{1-\lambda}=0.
    \end{align*}
where the last equality holds by definition of $C(n,r,\lambda)$.

It remains to show that $g$ changes sign at most once on $[0,\infty)$. Because $(1-F(x))^{-\lambda}f(x)\ge 0$, the sign of $g$ is uniquely determined by the sign of the function $A$.  To argue about the sign of $A$ we use the following lemma.

\begin{lemma}
    \label{lemma:unimodal-beta-average}
    The function $\alpha(t)=\frac{1}{t}B(t;a,b)$ with $a>1,b\ge 1$ is unimodal on the interval $[0,1]$ with a maximum attained at some $t^* \in (0,1)$. 
\end{lemma}

Because $1 - F(x)$ is non-increasing, the above lemma implies that also the mapping $x\mapsto A(1-F(x))$ is unimodal. This is equivalent to saying that the mapping $t \mapsto A(t)$ for $t \in (0,1]$ is unimodal, because $0 \leq 1 - F(x) \leq 1$. Together with the already established fact that $A(1) > 0$, this implies that the mapping $t \mapsto A(t)$ changes sign at most once, because in order for this mapping to change sign twice, it is necessary that $A(0) < 0$ and $A(1) < 0$.

To conclude, all conditions on $g$ in Lemma \ref{lemma:barlow} are satisfied as long as $r \leq r^*$, and so we have established \eqref{eq:C-lower-bound} for every such $r$, and therefore in particular for $r = r^*$.
\end{cleverproof}

\section{Conclusion}

In this paper we studied sample-based posted pricing relative to the personalised pricing benchmark given by the mean valuation. 
We fully resolve the problem for valuations with an increasing hazard rate by showing that posting the sample mean as the price is the optimal pricing rule among all (possibly randomised) pricing rules. 
For $\lambda$-regular valuation distributions we find the optimal deterministic pricing rule in the single-sample setting, and for the many-sample case we identify a simple order-statistic pricing rule and prove its asymptotic optimality.

Several interesting questions remain open. First, it would be interesting to identify the precise boundary between the regimes in which mean pricing and order-statistic pricing provide the better guarantee.
Determining (numerically) the optimal pricing policy for small $n$ and $\lambda>0$ remains open too. We do, however, suspect the improvements to be modest.
Finally, it would be interesting to extend our analysis to settings with multiple buyers, as well as to other models of partial information.

\newpage
\bibliographystyle{plainnat}
\bibliography{welfareReference}

%%%%%%%%%%%%%%%%%%%%%%%%%%%%%%%%%%%%%%%%%%%%%%%%%%%%%%%%%%%%
% Unlimited appendix

\newpage
\appendix

\section{Missing details Section \ref{sec:mhr}}
\label{appendix:eks-proof}

We start with a formal definition of pricing rules using Markov kernels.
In the main text, expressions of the form $\kappa((v_1,\dots,v_n))$ were intended to represent the distribution of a random posted-price conditional on $V_1=v_1,\dots,V_n=v_n$. However, since the event $\{V_1=v_1,\dots,V_n=v_n\}$ typically has zero measure, such conditional distributions are not defined in the classical sense. The appropriate framework in this setting is a Markov kernel, which allows us to define measures that depend measurably on a parameter.
\begin{definition}[Markov kernel; (Aliprantis and Border, 2006, Definition 19.11)] % Dit moet kennelijk hard coded anders een error. '\cite[Definition 19.11]{aliprantis2006infinite}' werkt niet
    Let $(\Omega_1,\mathcal{A}_1)$ and $(\Omega_2,\mathcal{A}_2)$ be measurable spaces. 
    A Markov kernel is a map $\kappa:\Omega_1\times\mathcal{A}_2\to[0,1]$ satisfying the following two properties.
    \begin{itemize}
        \item [(i)] For each $\omega_1\in\Omega_1$, the map $\kappa(\omega_1,\cdot):\mathcal{A}_2\to[0,1]$ is a probability measure.
        \item [(ii)] For each $A_2\in\mathcal{A}_2$, the map $\kappa(\cdot, A_2):\Omega_1\to[0,1]$ is $\mathcal{A}_1$-measurable.
    \end{itemize}
\end{definition}
Condition (i) guarantees that each parameter value induces a probability measure, while (ii) ensures measurability with respect to the parameter.
We always work with Borel $\sigma$-algebras, and denote by $\mathcal{K}(\Omega_1,\Omega_2)$ the set of all Markov kernels from $\Omega_1$ to $\Omega_2$.

With this framework, expectations are well-defined. Let $X$ be an $\Omega_1$-valued random variable with law $\nu$, let $\kappa\in \mathcal{K}(\Omega_1, \Omega_2)$, and denote by $Y_X$ an $\Omega_2$-valued random variable such that $\Prob(Y_{X}\in B|X=x)=\kappa(x, B)$, then for any integrable and measurable $f$
$$
\E[f(X,Y_X)]=\int_{\Omega_1}\int_{\Omega_2}f(x,y)\kappa(x,\mathrm{d}y)\mathrm{d}\nu(x).
$$
\medskip
\noindent We now continue with the upper bound proof of Theorem \ref{thm:mhr-pricing}.

\begin{proof}[Proof of Theorem \ref{thm:mhr-pricing} (upper bound of $(n/(n+1))^{n+1}$)]

Let $F_\theta(x)=1-e^{-\theta x}$ denote the distribution of the $\Expo(\theta)$ distribution. Note that it has a constant hazard rate, so $F_\theta\in\mathcal{P}_0$. Since $\{F_{\theta}:\theta>0\}\subset\mathcal{P}_0$ we have that
$$
Q_{n,0}\le\sup_{\kappa\in \mathcal{K}(\Rplus^n,\Rplus)}\inf_{F\in\{F_{\theta}:\theta>0\}}\frac{\E[\Rev(F,Y_{V_1,\dots,V_n})]}{\E[V]}.
$$
This sup-inf problem on the right hand side can be rewritten as the following infinite linear program.
    \begin{align*}
        &\sup && C\\
        &\text{subject to} &&\frac{\E[\Rev(F,Y_{Z_1\dots,Z_n})]}{\E[Z]}\ge C\text{ for all } \theta>0,
        \\
        &&& \kappa((z_1,\dots,z_n))\in\mathcal{P}(\mathbb{R}_+)\text{ for all }(z_1,\dots,z_n)\in\mathbb{R}_+^n\text{, } C\in\mathbb{R},
    \end{align*}
where $Z,Z_1,\dots,Z_n\sim \Expo(\theta)$.
    The pointwise constraint on $\kappa$ is simply the definition of $\kappa\in\mathcal{K}(\mathbb{R}_+^n,\mathbb{R}_+)$.
    The left hand side of the first expectation can be worked out
    \begin{align*}
        h(\theta,\kappa)&:=\frac{\E[\Rev(F,Y_{Z_1\dots,Z_n})]}{\E[Z]}=\theta \int_{\mathbb{R}_+^n}\prod_{i=1}^n\theta e^{-\theta z_i}\int_0^\infty se^{-\theta s} \kappa((z_1,\dots,z_n), \dd s)\dd \boldsymbol{z}\\
        &=\int_{\mathbb{R}^n_+}\int_0^\infty\theta^{n+1}se^{-\theta\left(s+\sum_{i=1}^nz_i\right)} \kappa(\boldsymbol{z},\dd s)\dd \boldsymbol{z}.
    \end{align*}
    The constraint then becomes $h(\theta,\kappa)\ge C$ for all $\theta>0$. The dual of the above LP can be written down directly, but as it is quite an intricate LP the derivation is provided.
    The Lagrangian of the problem is 
    $$\mathcal{L}(C,\kappa,\mu)=C+\int_0^\infty (h(\theta,\kappa)-C)\dd \mu(\theta)=C(1-\mu((0,\infty)))+\int_0^\infty h(\theta,\kappa)\dd \mu(\theta)$$
    with $\mu$ a non-negative measure on $(0,\infty)$. 
    The dual function is $$\phi(\mu)=\sup_{C\in\mathbb{R},\kappa\in\mathcal{K}(\mathbb{R}^n_+,\mathbb{R}_+)}\mathcal{L}(C,\kappa,\mu).$$
    By weak duality, $D :=\inf_{\mu\ge 0}\phi(\mu)$ is an upper bound on the primal problem. Observe that $\phi(\mu)=+\infty$ if $\mu((0,\infty))\ne1$, so the infimum can be taken over $\mu\in\mathcal{P}((0,\infty))$. Consequently,
    \begin{align*}
        D&=\inf_{\mu\in\mathcal{P}((0,\infty))}\sup_{\kappa\in\mathcal{K}(\mathbb{R}^n_+,\mathbb{R}_+)}\int_0^\infty 
        \int_{\mathbb{R}^n_+}\int_0^\infty\theta^{n+1}se^{-\theta\left(s+\sum_{i=1}^nz_i\right)} \kappa(\boldsymbol{z},\dd s)\dd \boldsymbol{z}
        \dd \mu(\theta)\\
        &=\inf_{\mu\in\mathcal{P}((0,\infty))} \sup_{\kappa\in\mathcal{K}(\mathbb{R}^n_+,\mathbb{R}_+)}\int_{\mathbb{R}^n_+}\int_0^\infty\int_0^\infty\theta^{n+1}se^{-\theta\left(s+\sum_{i=1}^nz_i\right)}\dd \mu(\theta) \kappa(\boldsymbol{z},\dd s)\dd \boldsymbol{z}&\text{(Tonelli)}\\
        &\le\inf_{\mu\in\mathcal{P}((0,\infty))} \int_{\mathbb{R}^n_+}\sup_{s\ge 0}\int_0^\infty\theta^{n+1}se^{-\theta\left(s+\sum_{i=1}^nz_i\right)}\dd \mu(\theta)\dd \boldsymbol{z}.
    \end{align*}
    We use pointwise maximisation and that $\kappa(\boldsymbol{z})$ is a probability measure in the second step.
    \footnote{Equality can be established, but is omitted here as only an upper bound on the primal is of interest.
    One would need to verify that the conditions of the Kuratowski and Ryll-Nardzewski measurable selection theorem hold
\cite[Theorem 18.13]{aliprantis2006infinite}.}
    With standard linearisation, the dual permits the following formulation:
    \begin{align*}
        &\inf && \int_{\mathbb{R}_+^n}\lambda(\boldsymbol{z})\dd \boldsymbol{z}\\
        &\text{subject to} &&
        \lambda(\boldsymbol{z})\ge\int_0^\infty \theta^{n+1}s\exp\left(-\theta\left(s+\sum_{i=1}^nz_i\right)\right)\dd \mu(\theta)\text{ for all }\boldsymbol{z}\in\mathbb{R}_+^n\text{ and }s\ge 0,
        \\
        &&& \mu\in\mathcal{P}((0,\infty))\text{, } \lambda:\mathbb{R}_+^n\to\mathbb{R}.
    \end{align*} 
    The infimum over all $\lambda$ is upper bounded by the infimum over all $\lambda$ that can be written as $\lambda(\boldsymbol{z})=\beta\left(\sum_{i=1}^nz_i\right)$. An upper bound on the dual is therefore
    \begin{align*}
        &\inf && \int_{\mathbb{R}_+^n}\beta\left(\sum_{i=1}^nz_i\right)\dd \boldsymbol{z}\\
        &\text{subject to} &&
        \beta(y)\ge\int_0^\infty \theta^{n+1}se^{-\theta(s+y)} \dd \mu(\theta)\text{ for all }y\ge 0\text{ and }s\ge 0,
        \\
        &&& \mu\in\mathcal{P}((0,\infty))\text{, }\beta:\mathbb{R}_+\to\mathbb{R}.
    \end{align*}
    The objective can be simplified.
    Define a measure $\nu$ on $\mathbb{R}_+$ by 
    $$
        \nu(A):=\mathrm{Leb}\left\{\boldsymbol{z}\in\mathbb{R}_+^n:\sum_{i=1}^nz_i\in A\right\}.
    $$
    A direct computation shows that for $y\ge 0$ that $\nu([0,y])=y^n/n!$ (see Lemma \ref{lemma:simplex-volume} below), so $\dd \nu(y)=\frac{y^{n-1}}{(n-1)!}\dd y$. The objective hence simplifies to
    \begin{equation}
        \label{eq:objective-expo-killing-dual}
        \int_{\mathbb{R}_+^n}\beta\left(\sum_{i=1}^nz_i\right)\dd \boldsymbol{z}=\int_0^\infty \beta(y)\dd \nu(y)=\frac{1}{(n-1)!}\int_0^\infty y^{n-1}\beta(y)\dd y.
    \end{equation}
    Now, consider the following dual feasible solution
    \begin{align}
        \label{eq:dual-solution-measure}
        \dd \mu_m(\theta)&=\frac{1}{\log m}\int_1^me^{-t\theta}\dd t\dd \theta\\
        \label{eq:dual-solution-constraint}
        \beta_m(y)&=\begin{cases}
            \frac{(n-1)!}{\log m},&\text{if }y\in[0,1),\\
            \frac{(n-1)!}{\log m}C_ny^{-n},&\text{if }y\in[1,m),\\
            \frac{n!}{\log m}\frac{m}{y^{n+1}},&\text{if }y\ge m.
        \end{cases}
    \end{align}
    First the feasibility of $(\mu_m,\beta_m)$ will be verified and then its associated objective will be computed. 
    To start, $\dd \mu_m(\theta)\ge 0$ and
    $$
        \mu_m((0,\infty))=\int_0^\infty\frac{1}{\log m}\int_1^me^{-t\theta}\dd t\dd \theta=\frac{1}{\log m}\int_1^m\int_0^\infty e^{-t\theta}\dd \theta\dd t=\frac{1}{\log m}\int_1^m\frac{1}{t}\dd t=1,
    $$
    so $\mu_m$ is a probability measure on $(0,\infty)$. Next the $\beta$ constraints. Define the right hand side of this constraint as
    \begin{align*}
        \xi_m(s,y)&:=\int_0^\infty \theta^{n+1}se^{-\theta(s+y)}\dd \mu_m(\theta)=\frac{1}{\log m}\int_0^\infty\theta^{n+1}se^{-\theta(s+y)} \int_1^me^{-t\theta}\dd t\dd \theta\\
        &=\frac{1}{\log m}\int_1^ms\int_0^\infty \theta^{n+1}e^{-\theta(s+y+t)}\dd \theta\dd t=\frac{(n+1)!}{\log m}\int_1^m\frac{s}{(s+y+t)^{n+2}}\dd t.
    \end{align*}
    Using definition of the Gamma function in the last step.
    Observe that $\beta_m(y)\ge \xi_m(s,y)$ for all $s, y\ge0$ ensures feasibility. This is established through three general bounds. 
    Observe that $\xi_m(s,y)$ is decreasing in $y$, with this observation we recover the bound
    \begin{align*}
        \xi_m(s,y)&\le\xi_m(s,0)=\frac{(n+1)!}{\log m}\int_1^m\frac{s}{(s+t)^{n+2}}\dd t\\
        &\le \frac{(n+1)!}{\log m}\int_1^\infty\frac{s}{(s+t)^{n+2}}\dd t=\frac{n!}{\log m}\frac{s}{(s+1)^{n+1}}\\
        &\le \frac{n!}{\log m}\frac{1/n}{(1/n+1)^{n+1}}&\text{(maximum at }s=1/n\text{)}\\
        &= \frac{(n-1)!}{\log m}\left(\frac{n}{n+1}\right)^{n+1}<\frac{(n-1)!}{\log m}.
    \end{align*}
    The next bound is for $y\in[1,m)$.
    \begin{align*}
        \xi_m(s,y)&\le \frac{(n+1)!}{\log m}\int_0^\infty\frac{s}{(s+y+t)^{n+2}}\dd t\\
        &=\frac{n!}{\log m}\frac{s}{(s+y)^{n+1}}\\
        &=\frac{n!}{\log m}\frac{y/n}{(y/n+y)^{n+1}}&\text{(maximum at }s=y/n\text{)}\\
        &=\frac{(n-1)!}{\log m}C_ny^{-n}.
    \end{align*}
    The last bound is for $y\ge m$.
    \begin{align*}
        \xi_m(s,y)&\le \frac{(n+1)!}{\log m}\int_1^m\sup_{w\ge0}\frac{w}{(w+y+t)^{n+2}}\dd t\\
        &=\frac{(n+1)!}{\log m}\int_1^m\frac{\frac{y+t }{n+1}}{(\frac{y+t }{n+1}+y+t)^{n+2}}\dd t&\text{(maximum at }w=\frac{y+t }{n+1}\text{)}\\
        &=\frac{n!}{\log m}\left(1+\frac{1}{n+1}\right)^{-(n+2)}\int_1^m(y+t)^{-(n+1)}\dd t\\
        &<\frac{n!}{\log m}\left(1+\frac{1}{n+1}\right)^{-(n+2)}\int_0^m\frac{1}{y^{n+1}}\dd t\\
        &=\frac{n!}{\log m}\left(1+\frac{1}{n+1}\right)^{-(n+2)}\frac{m}{y^{n+1}}<\frac{n!}{\log m}\frac{m}{y^{n+1}}.
    \end{align*}
    Using the first, second and third bound for the $y$-values $[0,1)$, $[1,m)$ and $[m,\infty)$
    respectively shows, indeed, that $(\mu_m,\beta_m)$ is feasible. Plugging $\beta_m$ into \eqref{eq:objective-expo-killing-dual} yields the objective
    \begin{align*}
        \frac{1}{(n-1)
        !}\int_0^\infty y^{n-1}\beta_m(y)\dd y&=\frac{1}{(n-1)!}\Bigg(\int_0^1y^{n-1}\frac{(n-1)!}{\log m}\dd y+\\& \int_1^my^{n-1}\frac{(n-1)!}{\log m}C_ny^{-n}\dd y+\int_m^\infty y^{n-1}\frac{n!}{\log m}\frac{m}{y^{n+1}}\dd y\Bigg)\\
        &=\frac{1}{\log m}\left(\int_0^1y^{n-1}\dd y+C_n\int_1^m\frac{1}{y}\dd y+nm\int_m^\infty\frac{1}{y^2}\dd y\right)\\
        &=C_n+\frac{n+1/n}{\log m}.    
    \end{align*}
    Lastly, taking $m\to\infty$ shows that the value of the dual LP is at most $C_n$. By weak duality the value of the primal LP is at most $C_n$, showing the optimality of $C_n$.
\end{proof}
%\subsection{Measure theory and Markov kernels}

Below we prove Lemma \ref{lemma:simplex-volume} used in the proof of Theorem \ref{thm:mhr-pricing-impossibility}.

\begin{lemma}\label{lemma:simplex-volume}
    Let $n\in\mathbb{N}$ and $y\ge 0$, then
    $$
        \mathrm{Leb}\left\{\boldsymbol{z}\in\mathbb{R}_+^n:\sum_{i=1}^nz_i\le y\right\}=\frac{y^n}{n!}.
    $$
\end{lemma}
\begin{proof}
    This is a standard fact, but the proof is included here for completeness. The proof proceeds by induction. The base case is $n=1$ with $\mathrm{Leb}\left\{z_1\in\mathbb{R}_+:z_1\le y\right\}=y=y^1/1!$. Assume the statement holds for $n-1$, then
    \begin{align*}
        \mathrm{Leb}\left\{\boldsymbol{z}\in\mathbb{R}_+^n:\sum_{i=1}^nz_i\le y\right\}&=\int_0^y
        \mathrm{Leb}\left\{(z_1,\dots,z_{n-1})\in\mathbb{R}_+^{n-1}:\sum_{i=1}^{n-1}z_i\le y-z_n\right\}\dd z_n\\
        &=\int_0^y\frac{(y-z_n)^{n-1}}{(n-1)!}\dd z_n=\frac{1}{(n-1)!}\int_0^ys^{n-1}\dd s=\frac{y^n}{n!}.
    \end{align*}
    Using the induction hypothesis in the second equality and the substitution $s=y-z_n$.
\end{proof}

\newpage

\section{Missing details from Section \ref{sec:lambda-single}}
%\section{Proof of Lemma \ref{lemma:special-function-monotonicity}}

The missing details from Section \ref{sec:lambda-single} are given here.
\subsection{Proof of Lemma \ref{lemma:scaling-reduction}}\label{appendix:scaling-proof}

This result is essentially a consequence of Proposition 1 of \citet{allouah2022pricing}. However, since their result is stated for randomised pricing rules and uses a different normalisation, we provide a simplified argument here for completeness.

First of all, by taking $\Psi(x)=\gamma x$ we obtain that
$$
\sup_{\Psi:\mathbb{R}\to\mathbb{R}}\inf_{F\in\mathcal{P}_\lambda}\frac{\E[\Rev(F,\Psi(V_1))]}{\E[V]}\ge \sup_{\gamma > 0}\inf_{F\in\mathcal{P}_\lambda}\frac{\E[\Rev(F,\gamma V_1)]}{\E[V]}.
$$
We now show that the inequality holds in the opposite direction.
Fix $\epsilon>0$, then there exists a $\Psi_1$ such that
$$
\inf_{F\in\mathcal{P}_\lambda}\frac{\E[\Rev(F,\Psi_1(V_1))]}{\E[V]}\ge \sup_{\Psi}\inf_{F\in\mathcal{P}_\lambda}\frac{\E[\Rev(F,\Psi(V_1))]}{\E[V]}-\epsilon.
$$
Fix an arbitrary $F_1\in\mathcal{P}_\lambda$ and define $F_n(s)=F(ns)$. Note that $F_n\in\mathcal{P}_\lambda$ and therefore that 
$$
\frac{\E[\Rev(F_n,\Psi_1(V_{1,n}))]}{\E[V_n]}\ge \sup_{\Psi}\inf_{F\in\mathcal{P}_\lambda}\frac{\E[\Rev(F,\Psi(V_1))]}{\E[V]}-\epsilon,
$$
where $V_n,V_{1,n}\sim F_n$. Observe that $V_n\overset{d}{=}\frac{1}{n}V$ and that $V_{1,n}\overset{d}{=}\frac{1}{n}V_1$ for $V,V_1\sim F_1$. Using the identity $\Rev(F_n,p)=\frac{1}{n}\Rev(F_1,np)$ we recover that 
$$
\frac{\E[\Rev(F_n,\Psi_1(V_{1,n}))]}{\E[V_n]}=\frac{\E[\Rev(F_1,n\Psi_1(V_1/n))]}{\E[V]}.
$$
By the regularity assumption there is a constant $\gamma\in\mathbb{R}$ such that
$$
\Psi_\infty(p):=\lim_{n\to\infty}n\Psi_1\left(\frac{p}{n}\right)=p\lim_{n\to\infty}\frac{n}{p}\Psi_1\left(\frac{p}{n}\right)=p\lim_{t\to 0}\frac{\Psi_1(t)}{t}=p\gamma.
$$
Define $X_n=\Rev(F_1,n\Psi_1(V_1/n))$. 
The map $p\mapsto\Rev(F_1,p)$ is continuous, so $X_n\to\Rev(F_1,\gamma V_1)$.
From $\Rev(F_1,p)\le \E[V]$ it follows that $|X_n|\le \E[V]$, so by the dominated convergence theorem,
$$
\lim_{n\to\infty}\frac{\E[\Rev(F_1,n\Psi_1(V_1/n))]}{\E[V]}=\frac{\E[\Rev(F_1,\gamma V_1)]}{\E[V]}.
$$
As the choice of $F_1$ was arbitrary, we have that
$$
\inf_{F_1\in\mathcal{P}_\lambda}\frac{\E[\Rev(F_1,\gamma V_1)]}{\E[V]}\ge \sup_{\Psi}\inf_{F\in\mathcal{P}_\lambda}\frac{\E[\Rev(F,\Psi(V_1))]}{\E[V]}-\epsilon
$$
for some $\gamma$. The inequality is preserved when taking the supremum over $\gamma$ finishing the proof. 

The assumption that $\lim_{t\to0}\Psi(t)/t$ exists is the same assumption under which \citet{allouah2022pricing} establish their result. This assumption can be weakened: it suffices to assume that
$$\lim_{t\to0}\frac{1}{t}\int_0^t\frac{\Psi(s)}{s}\dd s$$
exists, or, more generally, that the oscillations of $\Psi(t)/t$ admit a well-defined asymptotic average near the origin. For clarity, however, we work under the stronger assumption.

\subsection{Proof of Lemma \ref{lemma:special-function-monotonicity}}
\label{app:hyper-geometric}
Before proving Lemma \ref{lemma:special-function-monotonicity}, we first give a primer on hypergeometric functions and the properties of these functions that we need.
We refer to Abramowitz and Stegun \cite[Chapter 15]{abramowitz1948handbook} for a treaty of hypergeometric functions, which we lay out here as well.

For $a,b,c, \in \Rplus$ and $|z| < 1$, we consider the hypergeometric function 
$$
F(a,b,c,z) :=\;_2F_1(a,b;c;z) = \sum_{n=0}^\infty \frac{(a)_n (b)_n}{(c)_n} \frac{z^n}{n!} = 1 + \frac{ab}{c}\frac{z}{1!} + \frac{a(a+1)b(b+1)}{c(c+1)}\frac{z^2}{2!} + \cdots
$$
with $(q)_n$ the so-called (rising) Pochhammer symbol defined by
$$
(q)_n = \begin{cases}  1  & n = 0 \\
  q(q+1) \cdots (q+n-1) & n > 0.
 \end{cases}
$$
Euler gave the following integral representation for hypergeometric functions, provided that $c > b > 0$, that we will rely on to transform our integral representation of $I(\gamma)$ in Lemma \ref{lemma:special-function-monotonicity} into a hypergeometric function \cite[Equation 15.3.1]{abramowitz1948handbook}:
\begin{align}
    F(a,b,c,z)=\frac{\Gamma(c)}{\Gamma(b)\Gamma(c-b)}\int_0^1t^{b-1}(1-t)^{c-b-1}(1-tz)^{-a}\dd t,
    \label{eq:hyper-euler-type}
\end{align}
where $\Gamma(\cdot)$ is the gamma function defined by $\Gamma(y) = \int_0^\infty t^{z-1}e^{-t} \mathrm{d}t$.

To show that $I(\gamma)$ is increasing in $\gamma$, we will show that its derivative is positive. To do this we need to understand the derivative of a hypergeometric function. A property that follows directly from the definition of $F$, see \cite[Equation 15.2.1]{abramowitz1948handbook}, is that its derivative with respect to $z$ is given by
 \begin{align}
\frac{\dd }{\dd z} F(a,b,c,z)=\frac{ab}{c}F(a+1,b+1,c+1,z).
\label{eq:hyper-derivative}
\end{align}

Finally, the functions
$
_2F_1 (a \pm 1,b, c, z), \; {}_2F_1 (a,b \pm 1, c ,z), \;{}_2F_1 (a,b,c \pm 1,z)
$
are called contiguous to $F(a,b,c,z)$. Gauss showed that $F(a,b,c,z)$ can be written as a linear combination of any two of its contiguous functions, where the coefficients in this linear combination are rational functions in $a,b,c$ and $z$. Considering the two contiguous relations 
   \begin{align*}
        F\left(a,b,c+1,z\right)&=\frac{cF\left(a,b,c,z\right)-bF\left(a,b+1,c+1,z\right)}{\left(c-b\right)},\\
        F\left(a,b+1,c+1,z\right)&=\frac{\left(c-b\right)F\left(a,b,c+1,z\right)-a\left(1-z\right)F\left(a+1,b+1,c+1,z\right)}{\left(c-a-b\right)}.
    \end{align*}
which are \cite[Equation 15.2.24]{abramowitz1948handbook} and \cite[Equation 15.2.15]{abramowitz1948handbook}, respectively, and eliminating the common term $F\left(a,b+1,c+1,z\right)$ from them, gives the new relation
\begin{align}
 F(a+1,b+1,c+1,z)=\frac{(c-a)(c-b)F(a,b,c+1,z)-c(c-a-b)F(a,b,c,z)}{ab(1-z)}
       \label{eq:hyper-contiguous}
\end{align}
that will be used to write the derivative \eqref{eq:hyper-derivative} in terms of the hypergeometric functions $F(a,b,c+1,z)$ and $F(a,b,c,z)$ which are easier to work with for our purposes. In particular, we will rely on the observation  
\begin{align}
F(a,b,c,z) > F(a,b,c+1,z)
       \label{eq:hyper-c-monotone}
\end{align}
that follows directly from the definition of $F$. We can now proof the lemma.
\begin{cleverproof}{lemma:special-function-monotonicity}
    Recall that we need to show that for fixed $\lambda\in(0,1)$, the function 
    $$
    I(\gamma) = \frac{1-\lambda}{\lambda^2}\gamma\int_0^\infty v(1+\gamma v)^{-1/\lambda}(1+v)^{-1/\lambda-1}\dd v
    $$ is increasing on $(0,1)$.
 First, we rewrite $I(\gamma)$ into a product of $\gamma$ and a hypergeometric function by applying the substitution $t=1-1/(1+v)$ and \eqref{eq:hyper-euler-type}. Note that under the given substitution, we have $v=t/(1-t)$ and $\dd v=\frac{1}{(1-t)^2}\dd t$, so that the integral  becomes
    \begin{align*}
        I(\gamma)
        &=\frac{1-\lambda}{\lambda^2}\gamma\int_0^1 \frac{t}{1-t}\left(1+\gamma \frac{t}{1-t}\right)^{-1/\lambda}\left(1+\frac{t}{1-t}\right)^{-1/\lambda-1}\frac{1}{(1-t)^2}\dd t\\
        &=\frac{1-\lambda}{\lambda^2}\gamma\int_0^1 \frac{t}{(1-t)^3}\left( \frac{1+(\gamma-1)t}{1-t}\right)^{-1/\lambda}\left(\frac{1}{1-t}\right)^{-1/\lambda-1}\dd t\\
        &=\frac{1-\lambda}{\lambda^2}\gamma\int_0^1 t(1-t)^{\frac{2}{\lambda}-2}(1+(\gamma-1)t)^{-1/\lambda} \dd t.
    \end{align*}
    %By Lemma \ref{lemma:special-function-monotonicity} the above is above function is increasing in $\gamma$. This shows that $\gamma < 1$ is not optimal completing the final case.

Using the integral representation of Euler in \eqref{eq:hyper-euler-type}, with $a = 1/\lambda, b = 2, c = 1 + 2/\lambda$ and the observation that then $c > b > 0$, yields
$$
I(\gamma) = \frac{1-\lambda}{\lambda^2} \frac{\Gamma(2)\Gamma(\frac{2}{\lambda}-1)}{\Gamma(1+\frac{2}{\lambda})}\cdot\gamma F\left(\frac{1}{\lambda},2,1+\frac{2}{\lambda},1-\gamma\right).
$$
The factors that depend exclusively on $\lambda$ can be seen to be positive constants (as in particular the Gamma function is always positive), so to prove that $I(\gamma)$ is increasing it suffices to show that the function
$$
H(\gamma) = \gamma F\left(\frac{1}{\lambda},2,1+\frac{2}{\lambda},1-\gamma\right)
$$
is increasing. Note that the hypergeometric function is well-defined as $\gamma \in (0,1)$ and so $|1-\gamma| < 1$.
    To show that $H'(\gamma)>0$, observe that
    \begin{align}
    H'(\gamma) &=         \frac{\dd }{\dd \gamma}\gamma F\left(\frac{1}{\lambda},2,1+\frac{2}{\lambda},1-\gamma\right) \nonumber \\
        &=F\left(\frac{1}{\lambda},2,1+\frac{2}{\lambda},1-\gamma\right)-\frac{2\gamma}{2+\lambda} \frac{\dd }{\dd \gamma}F\left(\frac{1}{\lambda},2,1+\frac{2}{\lambda},1-\gamma\right) \nonumber \\
        &=F\left(\frac{1}{\lambda},2,1+\frac{2}{\lambda},1-\gamma\right)-\frac{2\gamma}{2+\lambda} F\left(1+\frac{1}{\lambda},3,2+\frac{2}{\lambda},1-\gamma\right) 
        \label{eq:H-derivative}
    \end{align}
    using the product rule in the second equality and \eqref{eq:hyper-derivative} in the third equality.
    Using the relation \eqref{eq:hyper-contiguous} for the second term in \eqref{eq:H-derivative}, recalling that $a = 1/\lambda, b = 2$ and $c = 1 + 2/\lambda$, gives
    \begin{align*}
        \frac{2\gamma}{2+\lambda} F\left(1+\frac{1}{\lambda},3,2+\frac{2}{\lambda},1-\gamma\right)&=\frac{(\lambda+1)(2-\lambda)}{\left(2+\lambda\right)\lambda}F\left(\frac{1}{\lambda},2,2+\frac{2}{\lambda},1-\gamma\right)\\&-\frac{1-\lambda}{\lambda}F\left(\frac{1}{\lambda},2,1+\frac{2}{\lambda},1-\gamma\right).    
    \end{align*}
    Plugging this into \eqref{eq:H-derivative}, and simplifying the coefficient in front of $F(1/\lambda,2,1 + 2/\lambda,2-\gamma)$, yields
    $$
    H'(\gamma)=\frac{1}{\lambda}F\left(\frac{1}{\lambda},2,1+\frac{2}{\lambda},1-\gamma\right)-\frac{(\lambda+1)(2-\lambda)}{\left(2+\lambda\right)\lambda}F\left(\frac{1}{\lambda},2,2+\frac{2}{\lambda},1-\gamma\right).
    $$
    Using, for $\lambda \in (0,1)$, that
    $$
    \frac{1}{\lambda}-\frac{(\lambda+1)(2-\lambda)}{\left(2+\lambda\right)\lambda}>0
    $$
    gives
    $$
    H'(\gamma) > \frac{(\lambda+1)(2-\lambda)}{\left(2+\lambda\right)\lambda} \left[ F\left(\frac{1}{\lambda},2,1+\frac{2}{\lambda},1-\gamma\right) - F\left(\frac{1}{\lambda},2,2+\frac{2}{\lambda},1-\gamma\right)\right].
    $$
    Finally, inequality \eqref{eq:hyper-c-monotone} in combination with $(\lambda+1)(2-\lambda) > ((2+\lambda)\lambda) > 0$ for all $\lambda \in (0,1)$, then implies $H'(\gamma) > 0$. This completes the proof.
\end{cleverproof}

\subsection{Numerical upper bound on \texorpdfstring{$Q_{1,\lambda}$}{Q 1,lambda}}\label{appendix:numerical-bounds}
The guarantee from Theorem \ref{thm:general-pricing-impossibility} is quite weak as the theorem is only meant to establish the optimal asymptotic rate.
Here we give a refined bound for $n=1$, which is also visualised in Figure \ref{fig:Q-1-bounds}.
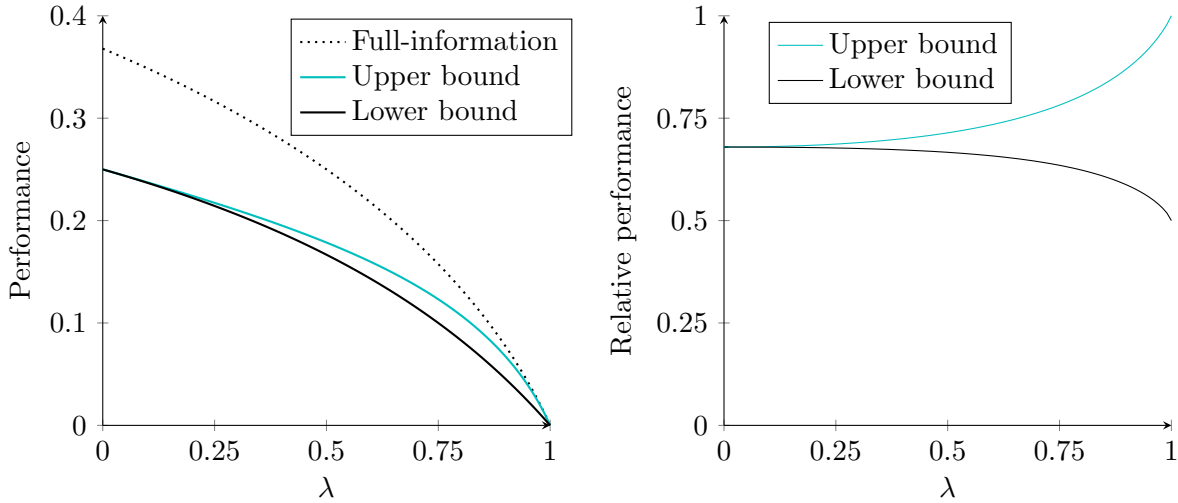
\begin{figure}[!h]
    \centering
    \hspace{-0.39cm}
    \begin{tikzpicture}
    \definecolor{cbblue}{rgb}{0.00,0.75,0.75}
    \definecolor{cbred}{rgb}{0.84,0.37,0.00}
    \definecolor{cborange}{rgb}{0.90,0.60,0.00}
    \definecolor{cbgreen}{rgb}{0.00,0.62,0.45}
    \definecolor{cbpurple}{rgb}{0.80,0.47,0.65}

    \begin{groupplot}[
        group style={
            group size=2 by 1,
            horizontal sep=2.3cm,
        },
        axis x line=bottom,
        axis y line=left,
        xlabel style={at={(axis description cs:0.5, -0.1)}, anchor=north},
        grid=none,
        width=7.5cm, height=7cm,
        legend style={
            at={(0.6,0.5)},
            anchor=west,
            cells={anchor=west}
        },
    ]
    
    % Plot 1 - log10(lambda)
    \nextgroupplot[
            xmin=0, xmax=1,
            ymin=0, ymax=0.4,
            axis lines=left,
            xlabel={$\lambda$},
            ylabel={Performance},
            xtick={0,0.25,0.5,0.75,1},
            ytick={0,0.1,0.2,0.3,0.4},
            enlargelimits=false,
            clip=true,
            legend style={
                at={(0.42,0.85)},
                anchor=west
            }
        ]
    \addplot[
                color=black, dotted,
                thick
            ] coordinates {
( 0 ,     0.3678794 )
( 0.01 ,  0.3660323 )
( 0.02 ,  0.3641697 )
( 0.03 ,  0.3622912 )
( 0.04 ,  0.3603967 )
( 0.05 ,  0.3584859 )
( 0.06 ,  0.3565586 )
( 0.07 ,  0.3546144 )
( 0.08 ,  0.3526532 )
( 0.09 ,  0.3506746 )
( 0.1 ,  0.3486784 )
( 0.11 ,  0.3466643 )
( 0.12 ,  0.344632 )
( 0.13 ,  0.3425812 )
( 0.14 ,  0.3405115 )
( 0.15 ,  0.3384227 )
( 0.16 ,  0.3363145 )
( 0.17 ,  0.3341864 )
( 0.18 ,  0.3320382 )
( 0.19 ,  0.3298695 )
( 0.2 ,  0.32768 )
( 0.21 ,  0.3254692 )
( 0.22 ,  0.3232369 )
( 0.23 ,  0.3209825 )
( 0.24 ,  0.3187058 )
( 0.25 ,  0.3164062 )
( 0.26 ,  0.3140835 )
( 0.27 ,  0.311737 )
( 0.28 ,  0.3093664 )
( 0.29 ,  0.3069713 )
( 0.3 ,  0.3045511 )
( 0.31 ,  0.3021053 )
( 0.32 ,  0.2996335 )
( 0.33 ,  0.2971351 )
( 0.34 ,  0.2946096 )
( 0.35 ,  0.2920564 )
( 0.36 ,  0.289475 )
( 0.37 ,  0.2868647 )
( 0.38 ,  0.284225 )
( 0.39 ,  0.2815553 )
( 0.4 ,  0.2788548 )
( 0.41 ,  0.2761229 )
( 0.42 ,  0.273359 )
( 0.43 ,  0.2705622 )
( 0.44 ,  0.2677319 )
( 0.45 ,  0.2648672 )
( 0.46 ,  0.2619674 )
( 0.47 ,  0.2590317 )
( 0.48 ,  0.2560591 )
( 0.49 ,  0.2530489 )
( 0.5 ,  0.25 )
( 0.51 ,  0.2469115 )
( 0.52 ,  0.2437824 )
( 0.53 ,  0.2406117 )
( 0.54 ,  0.2373982 )
( 0.55 ,  0.2341409 )
( 0.56 ,  0.2308386 )
( 0.57 ,  0.2274899 )
( 0.58 ,  0.2240938 )
( 0.59 ,  0.2206487 )
( 0.6 ,  0.2171534 )
( 0.61 ,  0.2136063 )
( 0.62 ,  0.210006 )
( 0.63 ,  0.2063507 )
( 0.64 ,  0.2026388 )
( 0.65 ,  0.1988686 )
( 0.66 ,  0.1950382 )
( 0.67 ,  0.1911455 )
( 0.68 ,  0.1871886 )
( 0.69 ,  0.1831652 )
( 0.7 ,  0.1790731 )
( 0.71 ,  0.1749098 )
( 0.72 ,  0.1706727 )
( 0.73 ,  0.166359 )
( 0.74 ,  0.1619658 )
( 0.75 ,  0.1574901 )
( 0.76 ,  0.1529285 )
( 0.77 ,  0.1482775 )
( 0.78 ,  0.1435331 )
( 0.79 ,  0.1386915 )
( 0.8 ,  0.1337481 )
( 0.81 ,  0.1286982 )
( 0.82 ,  0.1235366 )
( 0.83 ,  0.1182579 )
( 0.84 ,  0.1128558 )
( 0.85 ,  0.1073237 )
( 0.86 ,  0.1016543 )
( 0.87 ,  0.0958394 )
( 0.88 ,  0.08987006 )
( 0.89 ,  0.08373618 )
( 0.9 ,  0.07742637 )
( 0.91 ,  0.07092771 )
( 0.92 ,  0.06422535 )
( 0.93 ,  0.05730202 )
( 0.94 ,  0.0501373 )
( 0.95 ,  0.04270657 )
( 0.96 ,  0.03497941 )
( 0.97 ,  0.02691671 )
( 0.98 ,  0.01846533 )
( 0.99 ,  0.009545485 )
( 1 ,  0 )
            
            };
            \addlegendentry{Full-information}

            \addplot[
                color=cbblue, solid,
                thick
            ] coordinates {
( 0 ,  0.25 )
( 0.01 ,  0.2487484 )
( 0.02 ,  0.2474936 )
( 0.03 ,  0.2462353 )
( 0.04 ,  0.2449734 )
( 0.05 ,  0.2437078 )
( 0.06 ,  0.2424383 )
( 0.07 ,  0.2411648 )
( 0.08 ,  0.2398869 )
( 0.09 ,  0.2386046 )
( 0.1 ,  0.2373176 )
( 0.11 ,  0.2360258 )
( 0.12 ,  0.2347289 )
( 0.13 ,  0.2334267 )
( 0.14 ,  0.232119 )
( 0.15 ,  0.2308056 )
( 0.16 ,  0.2294861 )
( 0.17 ,  0.2281604 )
( 0.18 ,  0.2268282 )
( 0.19 ,  0.2254892 )
( 0.2 ,  0.2241431 )
( 0.21 ,  0.2227897 )
( 0.22 ,  0.2214287 )
( 0.23 ,  0.2200597 )
( 0.24 ,  0.2186824 )
( 0.25 ,  0.2172965 )
( 0.26 ,  0.2159016 )
( 0.27 ,  0.2144974 )
( 0.28 ,  0.2130835 )
( 0.29 ,  0.2116596 )
( 0.3 ,  0.2102252 )
( 0.31 ,  0.2087799 )
( 0.32 ,  0.2073232 )
( 0.33 ,  0.2058549 )
( 0.34 ,  0.2043743 )
( 0.35 ,  0.2028809 )
( 0.36 ,  0.2013744 )
( 0.37 ,  0.1998542 )
( 0.38 ,  0.1983196 )
( 0.39 ,  0.1967703 )
( 0.4 ,  0.1952055 )
( 0.41 ,  0.1936246 )
( 0.42 ,  0.1920271 )
( 0.43 ,  0.1904122 )
( 0.44 ,  0.1887793 )
( 0.45 ,  0.1871276 )
( 0.46 ,  0.1854564 )
( 0.47 ,  0.1837649 )
( 0.48 ,  0.1820522 )
( 0.49 ,  0.1803174 )
( 0.5 ,  0.1785597 )
( 0.51 ,  0.1767781 )
( 0.52 ,  0.1749716 )
( 0.53 ,  0.1731392 )
( 0.54 ,  0.1712796 )
( 0.55 ,  0.1693918 )
( 0.56 ,  0.1674745 )
( 0.57 ,  0.1655264 )
( 0.58 ,  0.1635462 )
( 0.59 ,  0.1615323 )
( 0.6 ,  0.1594834 )
( 0.61 ,  0.1573977 )
( 0.62 ,  0.1552736 )
( 0.63 ,  0.1531092 )
( 0.64 ,  0.1509027 )
( 0.65 ,  0.1486519 )
( 0.66 ,  0.1463547 )
( 0.67 ,  0.1440089 )
( 0.68 ,  0.1416118 )
( 0.69 ,  0.1391609 )
( 0.7 ,  0.1366534 )
( 0.71 ,  0.1340861 )
( 0.72 ,  0.131456 )
( 0.73 ,  0.1287595 )
( 0.74 ,  0.1259928 )
( 0.75 ,  0.1231519 )
( 0.76 ,  0.1202324 )
( 0.77 ,  0.1172296 )
( 0.78 ,  0.1141382 )
( 0.79 ,  0.1109528 )
( 0.8 ,  0.1076671 )
( 0.81 ,  0.1042744 )
( 0.82 ,  0.1007675 )
( 0.83 ,  0.09713815 )
( 0.84 ,  0.0933774 )
( 0.85 ,  0.08947531 )
( 0.86 ,  0.08542077 )
( 0.87 ,  0.08120136 )
( 0.88 ,  0.07680306 )
( 0.89 ,  0.07220999 )
( 0.9 ,  0.06740404 )
( 0.91 ,  0.06236435 )
( 0.92 ,  0.05706671 )
( 0.93 ,  0.05148265 )
( 0.94 ,  0.0455783 )
( 0.95 ,  0.03931257 )
( 0.96 ,  0.03263452 )
( 0.97 ,  0.02547891 )
( 0.98 ,  0.01775811 )
( 0.99 ,  0.009343859 )
( 1 ,  0 )
            };
            \addlegendentry{Upper bound}

                        \addplot[
                color=black, solid,
                thick
            ] coordinates {
( 0 ,  0.25 )
( 0.01 ,  0.2487437 )
( 0.02 ,  0.2474747 )
( 0.03 ,  0.2461929 )
( 0.04 ,  0.244898 )
( 0.05 ,  0.2435897 )
( 0.06 ,  0.242268 )
( 0.07 ,  0.2409326 )
( 0.08 ,  0.2395833 )
( 0.09 ,  0.2382199 )
( 0.1 ,  0.2368421 )
( 0.11 ,  0.2354497 )
( 0.12 ,  0.2340426 )
( 0.13 ,  0.2326203 )
( 0.14 ,  0.2311828 )
( 0.15 ,  0.2297297 )
( 0.16 ,  0.2282609 )
( 0.17 ,  0.226776 )
( 0.18 ,  0.2252747 )
( 0.19 ,  0.2237569 )
( 0.2 ,  0.2222222 )
( 0.21 ,  0.2206704 )
( 0.22 ,  0.2191011 )
( 0.23 ,  0.2175141 )
( 0.24 ,  0.2159091 )
( 0.25 ,  0.2142857 )
( 0.26 ,  0.2126437 )
( 0.27 ,  0.2109827 )
( 0.28 ,  0.2093023 )
( 0.29 ,  0.2076023 )
( 0.3 ,  0.2058824 )
( 0.31 ,  0.204142 )
( 0.32 ,  0.202381 )
( 0.33 ,  0.2005988 )
( 0.34 ,  0.1987952 )
( 0.35 ,  0.1969697 )
( 0.36 ,  0.195122 )
( 0.37 ,  0.1932515 )
( 0.38 ,  0.191358 )
( 0.39 ,  0.189441 )
( 0.4 ,  0.1875 )
( 0.41 ,  0.1855346 )
( 0.42 ,  0.1835443 )
( 0.43 ,  0.1815287 )
( 0.44 ,  0.1794872 )
( 0.45 ,  0.1774194 )
( 0.46 ,  0.1753247 )
( 0.47 ,  0.1732026 )
( 0.48 ,  0.1710526 )
( 0.49 ,  0.1688742 )
( 0.5 ,  0.1666667 )
( 0.51 ,  0.1644295 )
( 0.52 ,  0.1621622 )
( 0.53 ,  0.1598639 )
( 0.54 ,  0.1575342 )
( 0.55 ,  0.1551724 )
( 0.56 ,  0.1527778 )
( 0.57 ,  0.1503497 )
( 0.58 ,  0.1478873 )
( 0.59 ,  0.1453901 )
( 0.6 ,  0.1428571 )
( 0.61 ,  0.1402878 )
( 0.62 ,  0.1376812 )
( 0.63 ,  0.1350365 )
( 0.64 ,  0.1323529 )
( 0.65 ,  0.1296296 )
( 0.66 ,  0.1268657 )
( 0.67 ,  0.1240602 )
( 0.68 ,  0.1212121 )
( 0.69 ,  0.1183206 )
( 0.7 ,  0.1153846 )
( 0.71 ,  0.1124031 )
( 0.72 ,  0.109375 )
( 0.73 ,  0.1062992 )
( 0.74 ,  0.1031746 )
( 0.75 ,  0.1 )
( 0.76 ,  0.09677419 )
( 0.77 ,  0.09349593 )
( 0.78 ,  0.09016393 )
( 0.79 ,  0.08677686 )
( 0.8 ,  0.08333333 )
( 0.81 ,  0.07983193 )
( 0.82 ,  0.07627119 )
( 0.83 ,  0.07264957 )
( 0.84 ,  0.06896552 )
( 0.85 ,  0.06521739 )
( 0.86 ,  0.06140351 )
( 0.87 ,  0.05752212 )
( 0.88 ,  0.05357143 )
( 0.89 ,  0.04954955 )
( 0.9 ,  0.04545455 )
( 0.91 ,  0.0412844 )
( 0.92 ,  0.03703704 )
( 0.93 ,  0.03271028 )
( 0.94 ,  0.02830189 )
( 0.95 ,  0.02380952 )
( 0.96 ,  0.01923077 )
( 0.97 ,  0.01456311 )
( 0.98 ,  0.009803922 )
( 0.99 ,  0.004950495 )
( 1 ,  0 )
            
            };
            \addlegendentry{Lower bound}
    % \addlegendentry{$\lambda=0.95$}

    % Plot 2 - lambda
    \nextgroupplot[
            xmin=0, xmax=1,
            ymin=0, ymax=1,
            axis lines=left,
            xlabel={$\lambda$},
            ylabel={Relative performance},
            xtick={0,0.25,0.5,0.75,1},
            ytick={0,0.25,0.5,0.75,1},
            enlargelimits=false,
            clip=true,
            legend style={
                at={(0.1,0.885)},
                anchor=west
            }
        ]
        \addplot[
        color=cbblue,
        solid,
    ] coordinates {(0.01, 0.6796)
(0.00, 0.6796)
(0.01, 0.6796)
(0.02, 0.6796)
(0.03, 0.6797)
(0.04, 0.6797)
(0.05, 0.6798)
(0.06, 0.6799)
(0.07, 0.6801)
(0.08, 0.6802)
(0.09, 0.6804)
(0.1, 0.6806)
(0.11, 0.6808)
(0.12, 0.6811)
(0.13, 0.6814)
(0.14, 0.6817)
(0.15, 0.682)
(0.16, 0.6824)
(0.17, 0.6827)
(0.18, 0.6831)
(0.19, 0.6836)
(0.2, 0.684)
(0.21, 0.6845)
(0.22, 0.685)
(0.23, 0.6856)
(0.24, 0.6862)
(0.25, 0.6868)
(0.26, 0.6874)
(0.27, 0.6881)
(0.28, 0.6888)
(0.29, 0.6895)
(0.3, 0.6903)
(0.31, 0.6911)
(0.32, 0.6919)
(0.33, 0.6928)
(0.34, 0.6937)
(0.35, 0.6947)
(0.36, 0.6957)
(0.37, 0.6967)
(0.38, 0.6978)
(0.39, 0.6989)
(0.4, 0.7)
(0.41, 0.7012)
(0.42, 0.7025)
(0.43, 0.7038)
(0.44, 0.7051)
(0.45, 0.7065)
(0.46, 0.7079)
(0.47, 0.7094)
(0.48, 0.711)
(0.49, 0.7126)
(0.5, 0.7142)
(0.51, 0.716)
(0.52, 0.7177)
(0.53, 0.7196)
(0.54, 0.7215)
(0.55, 0.7235)
(0.56, 0.7255)
(0.57, 0.7276)
(0.58, 0.7298)
(0.59, 0.7321)
(0.6, 0.7344)
(0.61, 0.7369)
(0.62, 0.7394)
(0.63, 0.742)
(0.64, 0.7447)
(0.65, 0.7475)
(0.66, 0.7504)
(0.67, 0.7534)
(0.68, 0.7565)
(0.69, 0.7598)
(0.7, 0.7631)
(0.71, 0.7666)
(0.72, 0.7702)
(0.73, 0.774)
(0.74, 0.7779)
(0.75, 0.782)
(0.76, 0.7862)
(0.77, 0.7906)
(0.78, 0.7952)
(0.79, 0.8)
(0.8, 0.805)
(0.81, 0.8102)
(0.82, 0.8157)
(0.83, 0.8214)
(0.84, 0.8274)
(0.85, 0.8337)
(0.86, 0.8403)
(0.87, 0.8473)
(0.88, 0.8546)
(0.89, 0.8624)
(0.9, 0.8706)
(0.91, 0.8793)
(0.92, 0.8885)
(0.93, 0.8984)
(0.94, 0.9091)
(0.95, 0.9205)
(0.96, 0.933)
(0.97, 0.9466)
(0.98, 0.9617)
(0.99, 0.9789)
(1.00, 1.0000)
    };
    \addlegendentry{Upper bound}
    \addplot[
        color=black,
        solid,
    ] coordinates {
(0.00, 0.6796)
(0.01, 0.6796)
(0.02, 0.6796)
(0.03, 0.6795)
(0.04, 0.6795)
(0.05, 0.6795)
(0.06, 0.6795)
(0.07, 0.6794)
(0.08, 0.6794)
(0.09, 0.6793)
(0.1, 0.6793)
(0.11, 0.6792)
(0.12, 0.6791)
(0.13, 0.679)
(0.14, 0.6789)
(0.15, 0.6788)
(0.16, 0.6787)
(0.17, 0.6786)
(0.18, 0.6785)
(0.19, 0.6783)
(0.2, 0.6782)
(0.21, 0.678)
(0.22, 0.6778)
(0.23, 0.6777)
(0.24, 0.6775)
(0.25, 0.6772)
(0.26, 0.677)
(0.27, 0.6768)
(0.28, 0.6766)
(0.29, 0.6763)
(0.3, 0.676)
(0.31, 0.6757)
(0.32, 0.6754)
(0.33, 0.6751)
(0.34, 0.6748)
(0.35, 0.6744)
(0.36, 0.6741)
(0.37, 0.6737)
(0.38, 0.6733)
(0.39, 0.6728)
(0.4, 0.6724)
(0.41, 0.6719)
(0.42, 0.6714)
(0.43, 0.6709)
(0.44, 0.6704)
(0.45, 0.6698)
(0.46, 0.6693)
(0.47, 0.6687)
(0.48, 0.668)
(0.49, 0.6674)
(0.5, 0.6667)
(0.51, 0.6659)
(0.52, 0.6652)
(0.53, 0.6644)
(0.54, 0.6636)
(0.55, 0.6627)
(0.56, 0.6618)
(0.57, 0.6609)
(0.58, 0.6599)
(0.59, 0.6589)
(0.6, 0.6579)
(0.61, 0.6568)
(0.62, 0.6556)
(0.63, 0.6544)
(0.64, 0.6531)
(0.65, 0.6518)
(0.66, 0.6505)
(0.67, 0.649)
(0.68, 0.6475)
(0.69, 0.646)
(0.7, 0.6443)
(0.71, 0.6426)
(0.72, 0.6408)
(0.73, 0.639)
(0.74, 0.637)
(0.75, 0.635)
(0.76, 0.6328)
(0.77, 0.6305)
(0.78, 0.6282)
(0.79, 0.6257)
(0.8, 0.6231)
(0.81, 0.6203)
(0.82, 0.6174)
(0.83, 0.6143)
(0.84, 0.6111)
(0.85, 0.6077)
(0.86, 0.604)
(0.87, 0.6002)
(0.88, 0.5961)
(0.89, 0.5917)
(0.9, 0.5871)
(0.91, 0.5821)
(0.92, 0.5767)
(0.93, 0.5708)
(0.94, 0.5645)
(0.95, 0.5575)
(0.96, 0.5498)
(0.97, 0.541)
(0.98, 0.5309)
(0.99, 0.5186)
(1.00, 0.5000)
    };
    \addlegendentry{Lower bound}

    \end{groupplot}
    \end{tikzpicture}
    \caption{On the left absolute bounds for $Q_{1,\lambda}$ and the full-information performance $(1-\lambda)^{1/\lambda}$. On the right bounds relative to the full-information performance (right).}
    \label{fig:Q-1-bounds}
\end{figure}

\noindent{In Figure \ref{fig:Q-1-bounds} we can see that our sample-based pricing rule always attains at least half the performance when compared to the full-information setting. Moreover, the upper bound is particularly close to the lower bound for small values of $\lambda$.}

We now explain how the above upper bound is derived.
First, as shown in Lemma \ref{lemma:scaling-reduction}, we can restrict our focus to randomised scaling rules.\footnote{In Lemma \ref{lemma:scaling-reduction} we show this reduction for deterministic policies, but the argument extends to randomised policies too, as is done by \citet{allouah2022pricing} in Proposition 1.}
Consider the distributions functions $G_1$ and $G_2$ with $G_1(x)=1-\left(1+x\right)^{-1/\lambda}$ and $G_2(x)=\indicator{x\ge 1}$.
The distribution $G_2$ is strictly speaking not in $\mathcal{P}_\lambda$, but $G_{2,\alpha}(x)=1-(1+\lambda x^\alpha)^{-1/\lambda}$ is and $G_{2,\alpha}\to G_2$ pointwise, so any result for $G_2$ can be approximated within arbitrary precision.

Let the random scaling factor $S$ have distribution $\mu$. Under $G_2$, $V,V_1$ are $1$ with probability one, so
$$
\frac{\E[\Rev(G_2,SV_1)]}{\E[V]}=\E[S\indicator{S\le 1}].
$$
Let the function $I$ be as in the proof of Theorem \ref{thm:single-sample-det}, then under $G_1$ it is clear that
$$
\frac{\E[\Rev(G_1,SV_1)]}{\E[V]}=\E[I(S)].
$$
Consequently, we get the following upper bound
$$
Q_{1,\lambda}=\sup_{\mu\in\mathcal{P}(\Rplus)}\inf_{F\in\mathcal{P}_\lambda}\frac{\E[\Rev(F,SV_1)]}{\E[V]}\le \sup_{\mu\in\mathcal{P}(\Rplus)}\min\{\E[S\indicator{S\le 1}],\E[I(S)]\}.
$$
The above problem is very tractable. From Lemma \ref{lemma:special-function-monotonicity} we get that $I(1)\ge I(\gamma)$ for $\gamma\le 1$, so in optimality $\mu$ places no mass on $[0,1)$. Moreover, an optimal $\mu$ will only place mass at $1$ and $\gamma^*=\arg\max_{\gamma >0}I(\gamma)$. With this, the upper bound becomes 
$$
Q_{1,\lambda}\le \sup_{p\in[0,1]}\min\{p,I(1)p+(1-p)I(\gamma^*)\}=\frac{I(\gamma^*)}{1-I(1)+I(\gamma^*)}.
$$
The value of the above bound can be computed numerically.

\section{Non-optimality of the sample mean}\label{appendix:non-optimality-sample-mean}
For $n=2$, neither the sample mean policy nor the order statistic policy is optimal for all $\lambda \in (0,1)$ simultaneously. When $\lambda=0$, it follows from Theorem \ref{thm:mhr-pricing} that the sample mean policy achieves a performance of $(n/(n+1))^{n+1} = (2/3)^3 = 8/27$.
Theorem \ref{thm:order-statistic-policy} informs us that the best order statistic policy is to post the highest of the two samples ($r^*=2$), and that the associated performance is $C(2,2,0)=5/18<8/27$.

When $\lambda=1/2$, 
posting the highest price is still the best order statistic policy and it achieves a performance of $C(2,2,1/2)=1/5$. The performance ratio of the sample mean policy cannot be computed exactly in this case, however, any admissible distribution yields an upper bound. Below we show that the performance of the sample mean rule $\frac{1}{2}(V_1+V_2)$ is at most $0.19439$, which is strictly less than $1/5$ which is the performance guarantee of the order statistic $V_{2:2}$.

\begin{lemma}
Let $n = 2$, $\lambda = 1/2$, $F(x)=1-(1+x)^{-2}$ and $V,V_1,V_2\iid F$. Then
    $$
    \frac{\E\left[\Rev(F, \frac{1}{2}(V_1+V_2))\right]}{\E[V]}=\frac{136}{15}-\frac{64}{5}\log 2<0.19439.
    $$
    implying that 
    $$
    \inf_{F \in \mathcal{P}_{\lambda}} \frac{\E\left[\Rev(F, \frac{1}{2}(V_1+V_2))\right]}{\E[V]} \leq \frac{136}{15}-\frac{64}{5}\log 2 < \frac{1}{5} = \inf_{F \in \mathcal{P}_{\lambda}} \frac{\E\left[\Rev(F, V_{2:2})\right]}{\E[V]}.
    $$

\end{lemma}
\begin{proof}
    First, $\E[V]=\int_0^\infty(1-F(x))\dd x=1$. Next, 
    \begin{align*}
        I&:=\E\left[\Rev(F, \frac{1}{2}(V_1+V_2))\right]\\&=\int_0^\infty \int_0^\infty\frac{v_1+v_2}{2}\left(1-F\left(\frac{v_1+v_2}{2}\right)\right)\dd F(v_2)\dd F(v_1)\\
        &=\int_0^\infty \int_0^\infty v_1\left(1-F\left(\frac{v_1+v_2}{2}\right)\right)\dd F(v_2)\dd F(v_1)\\
        &=\int_0^\infty \int_0^\infty v_1\left(1+\frac{v_1+v_2}{2}\right)^{-2}\cdot2(1+v_2)^{-3}\cdot2(1+v_1)^{-3}\dd v_2\dd v_1\\
        &=16\int_0^\infty \frac{v_1}{(1+v_1)^3}\int_0^\infty \frac{1}{(2+v_1+v_2)^2(1+v_2)^3}\dd v_2\dd v_1.
    \end{align*}
    Using partial fractions we find that
    \begin{align*}
        &\;\;\;\;\;\int_0^\infty \frac{1}{(2+v_1+v_2)^2(1+v_2)^3}\dd v_2\\&=\int_0^\infty\Bigg(\frac{3}{\left(1+v_{1}\right)^{4}\left(1+v_{2}\right)}-\frac{2}{\left(1+v_{1}\right)^{3}\left(1+v_{2}\right)^{2}}+\frac{1}{\left(1+v_{1}\right)^{2}\left(1+v_{2}\right)^{3}}\\&\qquad\qquad-\frac{3}{\left(1+v_{1}\right)^{4}\left(2+v_{1}+v_{2}\right)}-\frac{1}{\left(1+v_{1}\right)^{3}\left(2+v_{1}+v_{2}\right)^{2}}\Bigg)\dd v_2\\
        &=\frac{3}{(1+v_1)^4}\int_0^\infty\left(\frac{1}{1+v_2}-\frac{1}{2+v_1+v_2}\right)\dd v_2
        -\frac{2}{\left(1+v_{1}\right)^{3}}\int_{0}^{\infty}\frac{1}{\left(1+v_{2}\right)^{2}}dv_{2}\\&+\frac{1}{\left(1+v_{1}\right)^{2}}\int_{0}^{\infty}\frac{1}{\left(1+v_{2}\right)^{3}}dv_{2}-\frac{1}{\left(1+v_{1}\right)^{3}}\int_{0}^{\infty}\frac{1}{\left(2+v_{1}+v_{2}\right)^{2}}\dd v_{2}\\
        &=\frac{3\log\left(2+v_{1}\right)}{\left(1+v_{1}\right)^{4}}-\frac{2}{\left(1+v_{1}\right)^{3}}+\frac{1}{2\left(1+v_{1}\right)^{2}}-\frac{1}{\left(1+v_{1}\right)^{3}}\frac{1}{\left(2+v_{1}\right)}.
    \end{align*}
    Substituting this result yields the expression
    \begin{align*}
        I=16\int_{0}^{\infty}\left(\frac{3v_{1}\log\left(2+v_{1}\right)}{\left(1+v_{1}\right)^{7}}-\frac{2v_{1}}{\left(1+v_{1}\right)^{6}}+\frac{v_{1}}{2\left(1+v_{1}\right)^{5}}-\frac{v_{1}}{\left(1+v_{1}\right)^{6}\left(2+v_{1}\right)}\right)dv_{1}.
    \end{align*}
    Using partial fractions, the last three terms can be computed
    \begin{align*}
        \int_{0}^{\infty}\frac{v_{1}}{\left(1+v_{1}\right)^{6}}dv_{1}&=\int_{0}^{\infty}\frac{1}{\left(1+v_{1}\right)^{5}}dv_{1}-\int_{0}^{\infty}\frac{1}{\left(1+v_{1}\right)^{6}}dv_{1}=\frac{1}{4}-\frac{1}{5}=\frac{1}{20}\\
        \int_{0}^{\infty}\frac{v_{1}}{\left(1+v_{1}\right)^{5}}dv_{1}&=\int_{0}^{\infty}\frac{1}{\left(1+v_{1}\right)^{4}}dv_{1}-\int_{0}^{\infty}\frac{1}{\left(1+v_{1}\right)^{5}}dv_{1}=\frac{1}{3}-\frac{1}{4}=\frac{1}{12}\\
        \int_{0}^{\infty}\frac{v_{1}}{\left(1+v_{1}\right)^{6}\left(2+v_{1}\right)}dv_{1}&=\int_1^\infty\frac{t-1}{t^6(1+t)}\dd t\\
        &=\int_1^\infty\left(-\frac{2}{1+t}+\frac{2}{t}-\frac{2}{t^{2}}+\frac{2}{t^{3}}-\frac{2}{t^{4}}+\frac{2}{t^{5}}-\frac{1}{t^{6}}\right)\dd t\\&=\int_1^\infty\left(\frac{2}{t}-\frac{2}{1+t}\right)\dd t-\frac{41}{30}=2\log 2-\frac{41}{30}
    \end{align*}
    For the first term use the substitutions $t=v_1/(1+v_1)$, $w=2-t$ and $u=1-t$ to obtain
    \begin{align*}
        \int_{0}^{\infty}\frac{v_1\log\left(2+v_1\right)}{\left(1+v_1\right)^{7}}\dd v_1&=\int_{0}^{1}t\left(1-t\right)^{4}\log\left(\frac{2-t}{1-t}\right)\dd t\\
        &=\int_{1}^{2}\left(2-w\right)\left(1-w\right)^{4}\log\left(w\right)dw-\int_{0}^{1}\left(1-u\right)u^{4}\log\left(u\right)du.
    \end{align*}
    Integration by parts informs us that
    $$
    \int t^k\log (t)\dd t=\frac{t^{k+1}}{k+1}\log(t)-\frac{1}{(k+1)^2} t^{k+1}+C.
    $$
    While cumbersome, with the above identity the exact value of the first term can be worked out to be $\frac{2}{5}\log2-\frac{89}{360}$. Finally,
    \begin{equation*}
        I=16\left(3\left(\frac{2}{5}\log2-\frac{89}{360}\right)-\frac{2}{20}+\frac{1}{2}\frac{1}{12}-\left(2\log 2-\frac{41}{30}\right)\right)=\frac{136}{15}-\frac{64}{5}\log2. \qedhere
    \end{equation*}
\end{proof}

\section{Proofs of lemmas in Section \ref{sec:lambda-multiple}}
\label{appendix:order-statistic-policy}

% \begin{lemma}
%     \label{lemma:constant-beta-equivalence}
%     It holds that
%     $$
%     B(n+2-r,r) > (1-\lambda)B(n+2-r-\lambda,r) \iff C(n,r,\lambda)>C(n,r-1,\lambda).
%     $$
% \end{lemma}
\begin{proof}[Proof of Lemma \ref{lemma:constant-beta-equivalence}]
    Note the identities 
    $$
    \binom{n}{r-1}=\frac{r}{n-r+1}\binom{n}{r}\text{ and }B(x+1,y-1)=\frac{x}{y-1}B(x,y).
    $$
    The result follows from rewriting
    \begin{align*}
        &\equalSpace C(n,r,\lambda)-C(n,r-1,\lambda)\\&=\frac{1-\lambda}{\lambda}r\binom{n}{r}(B(n+2-r-\lambda,r)-B(n+2-r,r))\\
        &-\frac{1-\lambda}{\lambda}(r-1)\binom{n}{r-1}(B(n+2-(r-1)-\lambda,r-1)-B(n+2-(r-1),r-1))\\
        &=\frac{1-\lambda}{\lambda}r\binom{n}{r}\Big(B(n+2-r-\lambda,r)-B(n+2-r,r)\\
        &-\frac{r-1}{n-r+1}(B(n+2-(r-1)-\lambda,r-1)-B(n+2-(r-1),r-1))\Big)\\
        &=\frac{1-\lambda}{\lambda}r\binom{n}{r}\Big(B(n+2-r-\lambda,r)-B(n+2-r,r)\\
        &-\frac{n+2-r-\lambda}{n-r+1}B(n+2-r-\lambda,r)-\frac{n+2-r}{n-r+1}B(n+2-r,r)\Big)\\
        &=\frac{1-\lambda}{\lambda}\frac{r}{n-r+1}\binom{n}{r}(B(n+2-r,r)-(1-\lambda)B(n+2-r-\lambda,r)).
    \end{align*}
    Now, $\frac{1-\lambda}{\lambda}\frac{r}{n-r+1}\binom{n}{r}\ge 0$ and the result follows. 
\end{proof}
\bigskip

% \begin{lemma}
%     \label{lemma:C-unimodality}
%     Let $n\ge 2$. The function $g:\{1,\dots,n\}\to \mathbb{R}$ given by $g(r)=C(n,r,\lambda)$ is unimodal.
% \end{lemma}
\begin{proof}[Proof of Lemma \ref{lemma:C-unimodality}]
    $g$ is unimodal if $s(r)=q(r)-q(r-1)$ has at most one sign change and if $s(2)>0$. First,
    \begin{align*}
        s(2)&=C(n,2,\lambda)-C(n,1,\lambda)\\
        &=\frac{1-\lambda}{\lambda}2\binom{n}{2}(B(n-\lambda,2)-B(n,2))-\frac{1-\lambda}{\lambda}\binom{n}{1}(B(n+1-\lambda,1)-B(n+1,1))\\
        &=\frac{1-\lambda}{\lambda}n\left(n-1\right)\left(\frac{1}{\left(n-\lambda\right)\left(n+1-\lambda\right)}-\frac{1}{n\left(n+1\right)}\right)-\frac{1-\lambda}{\lambda}n\left(\frac{1}{n+1-\lambda}-\frac{1}{n+1}\right)\\
        &=\frac{\left(1-\lambda\right)\left(n^{2}-n-1+\lambda\right)}{\left(n+1\right)\left(n-\lambda\right)\left(n+1-\lambda\right)}.
    \end{align*}
    For $n\ge 2$, all terms are positive, so $s(2)>0$. From Lemma \ref{lemma:constant-beta-equivalence} we obtain that
    \begin{align*}
        s(r)&=\frac{1-\lambda}{\lambda}\frac{r}{n-r+1}\binom{n}{r}(B(n+2-r,r)-(1-\lambda)B(n+2-r-\lambda,r))\\
        &=\frac{1-\lambda}{\lambda}\frac{r}{n-r+1}\binom{n}{r}\frac{\Gamma(r)\Gamma(n+2-r)}{\Gamma(n+2-\lambda)}\left(\frac{\Gamma(n+2-\lambda)}{\Gamma(n+2)}-(1-\lambda)\frac{\Gamma(n+2-\lambda-r)}{\Gamma(n+2-r)}\right).
    \end{align*}
    The last equality follows from the identity $B(a,b)=\Gamma(a)\Gamma(b)/\Gamma(a+b)$. Note that the sign of $s$ is the same as the sign of 
    $$
    \frac{\Gamma(n+2-\lambda)}{\Gamma(n+2)}-(1-\lambda)\frac{\Gamma(n+2-\lambda-r)}{\Gamma(n+2-r)}.
    $$
    The first term is a constant w.r.t. $r$ and the second term is decreasing in $r$. Therefore the sign of $s$ flips at most once, and the function $q$ is unimodal.
\end{proof}
\bigskip

\begin{cleverproof}{lemma:order-statistic-beta}
    First, let $A(x)=B(1-F(x); n-r+2,r)$. From the chain rule and $B(x;a;b)=\int_0^xt^{a-1}(1-t)^{b-1}\dd t$ we obtain
    $$
    A'(x)=-F(x)^{r-1}(1-F(x))^{n-r+1}f(x).
    $$
    Next, using the order statistic density formula 
    $$
    f_{V_{r:n}}(x) = \frac{n!}{(r-1)!(n-r)!} f_{V}(x) \left[ F_{V}(x) \right]^{r-1} \left[ 1 - F_{V}(x) \right]^{n-r},
    $$ 
    and integration by parts, we find that 
    \begin{align*}
        \E[\Rev(F,X_{r:n})]&=\E[X_{r:n}(1-F(X_{r:n}))]\\&=\int_0^\infty x(1-F(x))f_{X_{r:n}}(x)\dd x\\
        &=r\binom{n}{r}\int_0^\infty xF(x)^{r-1}(1-F(x))^{n-r+1}f(x)\dd x\\
        &=-r\binom{n}{r}\int_0^\infty xA'(x)\dd x\\
        &=\left[-r\binom{n}{r}xA(x)\right]_0^\infty+r\binom{n}{r}\int_0^\infty A(x)\dd x\\
        &=r\binom{n}{r}\int_0^\infty A(x)\dd x-r\binom{n}{r}\lim_{x\to\infty}xA(x).
    \end{align*}
    Showing that $xA(x)\to 0$ as $x\to\infty$ completes the proof. We have, using $(1-t)^{r} \leq 1$ in the second inequality, the following bounds on $A$:
    $$
    0\le A(x)=B(1-F(x),n-r+2,r)\le \frac{1}{n-r+2}(1-F(x))^{n-r+2}\le 1-F(x).
    $$
    We will be able to conclude the desired statement based on these bounds in combination with the following lemma.

    \begin{lemma}
    \label{lemma:tail-bound}
    Let $X$ be a non-negative random variable with distribution $F$ and finite mean, then $\lim_{x\to\infty}x(1-F(x))=0$.
\end{lemma}
\begin{proof}
    Note the following inequalities
    $$
    0\le x(1-F(x))=2\int_{x/2}^x(1-F(x))\dd t\le 2\int_{x/2}^x(1-F(t))\dd t,
    $$
    by the monotonicity of $F$.
    Because the mean is finite $\lim_{x\to\infty} \int_0^x(1-F(t))\dd t$ exists\\ and $\lim_{x\to \infty}\int_{x/2}^x(1-F(t))\dd t=0$. By the squeeze theorem, $\lim_{x\to\infty}x(1-F(x))=0$.
\end{proof}
    To  conclude, note that $\lim_{x\to\infty}xA(x)\le \lim_{x\to\infty}x(1-F(x))=0$ by Lemma \ref{lemma:tail-bound} so that then the claim  $xA(x)\to 0$ as $x\to\infty$ follows from the Squeeze theorem.
\end{cleverproof}
\bigskip

\begin{cleverproof}{lemma:const-gen-hazard-order-stat}
    First, because $F(x) = 1- (1+x)^{-1/\lambda}$ we have $\E[V]=\int_0^\infty(1-F(x))\dd x=\int_0^\infty(1+x)^{-1/
    \lambda}\dd x=\lambda/(1-\lambda)$. 
     We next carry out the substitution $s=1-F(x) = (1+x)^{-1/\lambda}$, i.e., $x = s^{-\lambda} - 1$ so that $\mathrm{d}x = -\lambda s^{-\lambda-1}\mathrm{d}s$. Together with Lemma \ref{lemma:order-statistic-beta}  the result follows by observing that
    \begin{align*}
        \E[\Rev(F,V_{r:n})]&=r\binom{n}{r}\int_0^\infty B(1-F(x); n-r+2,r)\dd x\\
        &=r\binom{n}{r}\int_1^{0} B(s; n-r+2,r)(-\lambda s^{-\lambda -1})\dd s\\
        &=\lambda r\binom{n}{r}\int_0^{1} \int_0^{s}t^{n-r+1}(1-t)^{r-1}s^{-\lambda -1}\dd t \dd s \\
        &=\lambda r\binom{n}{r}\int_0^{1} t^{n-r+1}(1-t)^{r-1}\int_t^{1}s^{-\lambda -1}\dd s \dd t \\
        &= r\binom{n}{r}\int_0^{1} t^{n-r+1}(1-t)^{r-1}(t^{-\lambda}-1) \dd t \\
        &=r\binom{n}{r}\left(B(n+2-r-\lambda, r)-B(n+2-r, r)\right).\qedhere
    \end{align*}
\end{cleverproof}

\bigskip

\begin{proof}[Proof of Lemma \ref{lemma:barlow}]
    Let $c$ denote the sign change location. For $x\in[0,c]$, $h(x)g(x)\ge h(c)g(x)$ since $h(x)\ge h(c)$ and $g(x)\ge 0$. For $x\in(c,\infty)$, $h(x)g(x)\ge h(c)g(x)$ since $h(c)\ge h(x)$ and $g(x)\le 0$. Therefore
    $$
    \int_0^\infty h(x)g(x)\dd x\ge \int_0^\infty h(c)g(x)\dd x=h(c)\int_0^\infty g(x)\dd x=0,
    $$
    which completes the proof.
\end{proof}

\bigskip

\begin{proof}[Proof of Lemma \ref{lemma:unimodal-beta-average}]
    If $b=1$, then $\alpha(t)=t^{a-1}/a$ which is increasing and therefore unimodal. If $b>1$ we have that
    $$
    \alpha'(t)=\frac{t\cdot t^{a-1}(1-t)^{b-1}-B(t;a,b)}{t^2}.
    $$
    Let $D(t)=tt^{a-1}(1-t)^{b-1}-B(t;a,b)$, then
    \begin{align*}
        D'(t)&=at^{a-1}(1-t)^{b-1}-(b-1)t^a(1-t)^{b-2}-t^{a-1}(1-t)^{b-1}\\&=t^{a-1}(1-t)^{b-2}\left(\left(a-1\right)\left(1-t\right)-\left(b-1\right)t\right).
    \end{align*}
    Note that $\left(a-1\right)\left(1-t\right)-\left(b-1\right)t=0$ implies $t=\frac{a-1}{b-1+a-1}=:t^*\in(0,1)$. Therefore, for $t\in(0,t^*)$ the sign of $D'$ is positive and for $t\in(t^*,1)$ the sign is negative. 
    Since $D'$ changes sign once, $D$ is increasing on $(0,t^*)$ and decreasing on $(t^*,1)$ and therefore has a unique maximum at $t^*$. Together with $D(0)=0$ and $D(1)<0$ this implies that $D$ has one zero in $(t^*,1)$ and a corresponding sign change from positive to negative. The sign of $\alpha'$ is the same as the sign of $D$, so $\alpha'$ has one sign change form positive to negative and $\alpha$ is therefore unimodal on $(0,1)$. Note that unimodality is preserved under shifting and positive scaling.
\end{proof}

\newpage
\section{Numerical and asymptotic results of Section \ref{sec:lambda-multiple}}\label{appendix:lambda-many-samples-asymptotic}
To start, we establish the equivalence of the two representations of the function $C$ in \eqref{eq:C-n-r-lambda}.
\begin{align*}
    C(n,r,\lambda)&=\frac{1-\lambda}{\lambda}r\binom{n}{r}\left(B(n+2-r-\lambda, r)-B(n+2-r, r)\right)\\
    &=\frac{1-\lambda}{\lambda}r\frac{\Gamma(n+1)}{\Gamma(r+1)\Gamma(n-r+1)}\left(\frac{\Gamma(n+2-r-\lambda)\Gamma(r)}{\Gamma(n+2+\lambda)}-\frac{\Gamma(n+2-r)\Gamma(r)}{\Gamma(n+2)}\right)\\
    &=\frac{1-\lambda}{\lambda}\frac{\Gamma(n+1)}{\Gamma(n-r+1)}\left(\frac{\Gamma(n+2-r-\lambda)}{\Gamma(n+2+\lambda)}-\frac{\Gamma(n+2-r)}{\Gamma(n+2)}\right)\\
    &=\frac{1-\lambda}{\lambda}\left(\frac{\Gamma(n+1)}{\Gamma(n-r+1)}\frac{\Gamma(n+2-r-\lambda)}{\Gamma(n+2+\lambda)}-\frac{n-r+1}{n+1}\right)\\
    &=\frac{1-\lambda}{\lambda}\left(\prod_{k=n-r+1}^n\frac{k}{k+1-\lambda}-\frac{n-r+1}{n+1}\right).
\end{align*}
With this form, $C(n,r^*,\lambda)$ can be calculated quite easily, and Corollary \ref{corol:C-n-r-lambda-asymp} can be proven.
\begin{proof}[Proof of Corollary \ref{corol:C-n-r-lambda-asymp}]
    \citet[Equation 6.1.47]{abramowitz1948handbook} inform us that 
    $$
    \frac{\Gamma(z+a)}{\Gamma(z+b)}=z^{a-b}\left(1+\frac{(a-b)(a+b-1)}{2z}+O\left(z^{-2}\right)\right),
    $$
    as $z\to\infty$. Let $s=n-r+1$ and $y=s/n$. Applying the above result twice yields
    \begin{align*}
        \prod_{k=yn}^n\frac{k}{k+1-\lambda}&=\frac{\Gamma(yn+1-\lambda)}{\Gamma(yn)}\frac{\Gamma(n+1)}{\Gamma(n+2-\lambda)}\\&=(yn)^{1-\lambda}\left(1-\frac{(1-\lambda)\lambda}{2yn}+O\left(n^{-2}\right)\right)n^{\lambda-1}\left(1-\frac{(1-\lambda)(2-\lambda)}{2n}+O\left(n^{-2}\right)\right)\\
        &=y^{1-\lambda}\left(1-\frac{(1-\lambda)\lambda}{2yn}-\frac{(1-\lambda)(2-\lambda)}{2n}+O\left(n^{-2}\right)\right)\\
        &=y^{1-\lambda}-\frac{1-\lambda}{2n}\left(\lambda y^{-\lambda}+(2-\lambda)y^{1-\lambda}\right)+O\left(n^{-2}\right).
    \end{align*}
    Moreover, $s/(n+1)=y-y/n+O(n^{-2})$. Let $g_n(y)=C(n,1+(1-y)n,\lambda)$, using the above we recover that
    $$
    g_n(y)=\frac{1-\lambda}{\lambda}(y^{1-\lambda}-y)+\frac{1-\lambda}{\lambda}\left(y-\frac{1-\lambda}{2}(2-\lambda)y^{1-\lambda}-\frac{1-\lambda}{2}\lambda y^{-\lambda}\right)\frac{1}{n}+O\left(n^{-2}\right).
    $$
    We have that $g_n(y)\to g(y)=\frac{1-\lambda}{\lambda}(y^{1-\lambda}-y)$ as $n\to\infty$. The function $g$ has a unique maximiser at $y^*=(1-\lambda)^{1/\lambda}$, so the maximiser of $g_n$ converges to $y^*$. For the original problem we can therefore conclude that $r^*/n\to 1-y^*=1-(1-\lambda)^{1/\lambda}$. Evaluating $g_n$ at $y^*$ yields 
    $$
    g_n(y^*)=(1-\lambda)^{1/\lambda}-\frac{1}{2}(1-\lambda)(1-(1-\lambda)^{1/\lambda})\frac{1}{n}+O\left(n^{-2}\right).
    $$
    As $g_n(y^*)$ is a surrogate of $C(n,r^*,\lambda)$ we can conclude that
    \begin{equation*}
        C(n,r^*,\lambda)=(1-\lambda)^{\frac{1}{\lambda}}-\frac{1}{2}(1-\lambda)(1-(1-\lambda)^{\frac{1}{\lambda}})\frac{1}{n}+O\left(n^{-2}\right)=(1-\lambda)^{\frac{1}{\lambda}}-\Theta\left(\frac{1}{n}\right).\qedhere
    \end{equation*}
\end{proof}
In Figure \ref{fig:performance} the functions $C(n,r^*,\lambda)$ and 
$$
R(n,\lambda)=n\frac{(1-\lambda)^{1/\lambda}-C(n,r^*,\lambda)}{\frac{1}{2}(1-\lambda)(1-(1-\lambda)^{\frac{1}{\lambda}})}
$$
are plotted for several $\lambda$ values. The function $R$ represents how accurate the asymptotic first order constant is. The wavy of these graphs is due to $r^*$ having to be an integer.

\begin{figure}[!h]
    \centering
    \hspace{-0.39cm}
    \begin{tikzpicture}
    \definecolor{cbblue}{rgb}{0.00,0.75,0.75}
    \definecolor{cbred}{rgb}{0.84,0.37,0.00}
    \definecolor{cborange}{rgb}{0.90,0.60,0.00}
    \definecolor{cbgreen}{rgb}{0.00,0.62,0.45}
    \definecolor{cbpurple}{rgb}{0.80,0.47,0.65}
    \begin{groupplot}[
        group style={
            group size=2 by 1,
            horizontal sep=2.3cm,
        },
        axis x line=bottom,
        axis y line=left,
        xlabel style={at={(axis description cs:0.5, -0.1)}, anchor=north},
        grid=none,
        width=7.5cm, height=7cm,
        legend style={
            at={(0.6,0.5)},
            anchor=west,
            cells={anchor=west}
        },
    ]
    
    % Plot 1 - log10(lambda)
    \nextgroupplot[
        ylabel={$C(n,r^*,\lambda)$},
        ylabel style={at={(axis description cs:-0.15,0.5)}},
        xlabel={$n$},
        xlabel style={
            at={(current axis.right of origin)},
            anchor=west
        },
        % Small work around, the limits are quirky
        ymin=0, ymax=0.33,
        ytick={0.001, 0.1, 0.2, 0.3},
        yticklabels={0, 0.1, 0.2, 0.3},
        xmin=0, xmax=100,
        xtick={25, 50, 75, 99.995},
        xticklabels={25, 50, 75, 100}
    ]
    \addplot[
        color=cbblue,
        solid,
    ] coordinates {
    (1, 0.2143)
(2, 0.2468)
(3, 0.2474)
(4, 0.2699)
(5, 0.2781)
(6, 0.279)
(7, 0.286)
(8, 0.29)
(9, 0.2908)
(10, 0.2938)
(11, 0.2963)
(12, 0.297)
(13, 0.2983)
(14, 0.3001)
(15, 0.3007)
(16, 0.3014)
(17, 0.3027)
(18, 0.3032)
(19, 0.3035)
(20, 0.3046)
(21, 0.305)
(22, 0.3051)
(23, 0.306)
(24, 0.3064)
(25, 0.3064)
(26, 0.3071)
(27, 0.3075)
(28, 0.3075)
(29, 0.308)
(30, 0.3083)
(31, 0.3084)
(32, 0.3087)
(33, 0.309)
(34, 0.3091)
(35, 0.3093)
(36, 0.3096)
(37, 0.3097)
(38, 0.3098)
(39, 0.3101)
(40, 0.3102)
(41, 0.3103)
(42, 0.3105)
(43, 0.3107)
(44, 0.3107)
(45, 0.3109)
(46, 0.311)
(47, 0.3111)
(48, 0.3112)
(49, 0.3114)
(50, 0.3114)
(51, 0.3115)
(52, 0.3116)
(53, 0.3117)
(54, 0.3118)
(55, 0.3119)
(56, 0.3119)
(57, 0.312)
(58, 0.3121)
(59, 0.3122)
(60, 0.3122)
(61, 0.3123)
(62, 0.3124)
(63, 0.3124)
(64, 0.3125)
(65, 0.3126)
(66, 0.3126)
(67, 0.3127)
(68, 0.3127)
(69, 0.3127)
(70, 0.3128)
(71, 0.3129)
(72, 0.3129)
(73, 0.313)
(74, 0.313)
(75, 0.313)
(76, 0.3131)
(77, 0.3131)
(78, 0.3132)
(79, 0.3132)
(80, 0.3133)
(81, 0.3133)
(82, 0.3133)
(83, 0.3134)
(84, 0.3134)
(85, 0.3134)
(86, 0.3135)
(87, 0.3135)
(88, 0.3135)
(89, 0.3136)
(90, 0.3136)
(91, 0.3136)
(92, 0.3137)
(93, 0.3137)
(94, 0.3137)
(95, 0.3137)
(96, 0.3138)
(97, 0.3138)
(98, 0.3138)
(99, 0.3139)
(100, 0.3139)
    };
    % \addlegendentry{$\lambda=0.25$}
    
    \addplot[
        color=black,
        solid,
    ] coordinates {
    (1, 0.1667)
(2, 0.2)
(3, 0.2071)
(4, 0.2095)
(5, 0.2208)
(6, 0.2258)
(7, 0.2274)
(8, 0.2283)
(9, 0.2321)
(10, 0.234)
(11, 0.2347)
(12, 0.2352)
(13, 0.237)
(14, 0.2381)
(15, 0.2385)
(16, 0.2387)
(17, 0.2399)
(18, 0.2405)
(19, 0.2407)
(20, 0.2409)
(21, 0.2417)
(22, 0.2421)
(23, 0.2423)
(24, 0.2424)
(25, 0.2429)
(26, 0.2432)
(27, 0.2434)
(28, 0.2435)
(29, 0.2439)
(30, 0.2441)
(31, 0.2442)
(32, 0.2443)
(33, 0.2446)
(34, 0.2447)
(35, 0.2448)
(36, 0.2449)
(37, 0.2451)
(38, 0.2453)
(39, 0.2453)
(40, 0.2454)
(41, 0.2456)
(42, 0.2457)
(43, 0.2458)
(44, 0.2458)
(45, 0.246)
(46, 0.2461)
(47, 0.2461)
(48, 0.2461)
(49, 0.2463)
(50, 0.2464)
(51, 0.2464)
(52, 0.2464)
(53, 0.2466)
(54, 0.2466)
(55, 0.2467)
(56, 0.2467)
(57, 0.2468)
(58, 0.2469)
(59, 0.2469)
(60, 0.2469)
(61, 0.247)
(62, 0.2471)
(63, 0.2471)
(64, 0.2471)
(65, 0.2472)
(66, 0.2472)
(67, 0.2473)
(68, 0.2473)
(69, 0.2473)
(70, 0.2474)
(71, 0.2474)
(72, 0.2474)
(73, 0.2475)
(74, 0.2475)
(75, 0.2475)
(76, 0.2476)
(77, 0.2476)
(78, 0.2476)
(79, 0.2477)
(80, 0.2477)
(81, 0.2477)
(82, 0.2478)
(83, 0.2478)
(84, 0.2478)
(85, 0.2478)
(86, 0.2479)
(87, 0.2479)
(88, 0.2479)
(89, 0.2479)
(90, 0.248)
(91, 0.248)
(92, 0.248)
(93, 0.248)
(94, 0.248)
(95, 0.2481)
(96, 0.2481)
(97, 0.2481)
(98, 0.2481)
(99, 0.2481)
(100, 0.2481)
    };
    % \addlegendentry{$\lambda=0.50$}

    \addplot[
        color=cbblue,
        dashed,
    ] coordinates {
    (1, 0.1)
(2, 0.1259)
(3, 0.1355)
(4, 0.1393)
(5, 0.1406)
(6, 0.1407)
(7, 0.1439)
(8, 0.1463)
(9, 0.1477)
(10, 0.1486)
(11, 0.149)
(12, 0.1491)
(13, 0.1497)
(14, 0.1506)
(15, 0.1512)
(16, 0.1516)
(17, 0.1518)
(18, 0.1519)
(19, 0.152)
(20, 0.1525)
(21, 0.1528)
(22, 0.1531)
(23, 0.1532)
(24, 0.1533)
(25, 0.1533)
(26, 0.1536)
(27, 0.1538)
(28, 0.154)
(29, 0.1541)
(30, 0.1541)
(31, 0.1541)
(32, 0.1543)
(33, 0.1544)
(34, 0.1545)
(35, 0.1546)
(36, 0.1547)
(37, 0.1547)
(38, 0.1547)
(39, 0.1549)
(40, 0.155)
(41, 0.155)
(42, 0.1551)
(43, 0.1551)
(44, 0.1551)
(45, 0.1552)
(46, 0.1553)
(47, 0.1553)
(48, 0.1554)
(49, 0.1554)
(50, 0.1554)
(51, 0.1554)
(52, 0.1555)
(53, 0.1556)
(54, 0.1556)
(55, 0.1556)
(56, 0.1556)
(57, 0.1556)
(58, 0.1557)
(59, 0.1558)
(60, 0.1558)
(61, 0.1558)
(62, 0.1558)
(63, 0.1558)
(64, 0.1559)
(65, 0.1559)
(66, 0.1559)
(67, 0.156)
(68, 0.156)
(69, 0.156)
(70, 0.156)
(71, 0.156)
(72, 0.1561)
(73, 0.1561)
(74, 0.1561)
(75, 0.1561)
(76, 0.1561)
(77, 0.1561)
(78, 0.1562)
(79, 0.1562)
(80, 0.1562)
(81, 0.1562)
(82, 0.1562)
(83, 0.1562)
(84, 0.1563)
(85, 0.1563)
(86, 0.1563)
(87, 0.1563)
(88, 0.1563)
(89, 0.1563)
(90, 0.1563)
(91, 0.1564)
(92, 0.1564)
(93, 0.1564)
(94, 0.1564)
(95, 0.1564)
(96, 0.1564)
(97, 0.1564)
(98, 0.1564)
(99, 0.1564)
(100, 0.1564)
    };
    % \addlegendentry{$\lambda=0.75$}
    
    \addplot[
        color=black,
        dashed,
    ] coordinates {
    (1, 0.0238)
(2, 0.0314)
(3, 0.0349)
(4, 0.037)
(5, 0.0383)
(6, 0.0391)
(7, 0.0397)
(8, 0.0402)
(9, 0.0405)
(10, 0.0408)
(11, 0.041)
(12, 0.0411)
(13, 0.0412)
(14, 0.0413)
(15, 0.0414)
(16, 0.0414)
(17, 0.0415)
(18, 0.0415)
(19, 0.0415)
(20, 0.0415)
(21, 0.0416)
(22, 0.0416)
(23, 0.0416)
(24, 0.0416)
(25, 0.0417)
(26, 0.0418)
(27, 0.0418)
(28, 0.0419)
(29, 0.0419)
(30, 0.042)
(31, 0.042)
(32, 0.042)
(33, 0.042)
(34, 0.0421)
(35, 0.0421)
(36, 0.0421)
(37, 0.0421)
(38, 0.0421)
(39, 0.0421)
(40, 0.0421)
(41, 0.0421)
(42, 0.0422)
(43, 0.0422)
(44, 0.0422)
(45, 0.0422)
(46, 0.0422)
(47, 0.0422)
(48, 0.0422)
(49, 0.0422)
(50, 0.0422)
(51, 0.0422)
(52, 0.0423)
(53, 0.0423)
(54, 0.0423)
(55, 0.0423)
(56, 0.0423)
(57, 0.0423)
(58, 0.0423)
(59, 0.0423)
(60, 0.0423)
(61, 0.0423)
(62, 0.0423)
(63, 0.0423)
(64, 0.0423)
(65, 0.0423)
(66, 0.0423)
(67, 0.0423)
(68, 0.0423)
(69, 0.0424)
(70, 0.0424)
(71, 0.0424)
(72, 0.0424)
(73, 0.0424)
(74, 0.0424)
(75, 0.0424)
(76, 0.0424)
(77, 0.0424)
(78, 0.0424)
(79, 0.0424)
(80, 0.0424)
(81, 0.0424)
(82, 0.0424)
(83, 0.0424)
(84, 0.0424)
(85, 0.0424)
(86, 0.0424)
(87, 0.0424)
(88, 0.0424)
(89, 0.0424)
(90, 0.0424)
(91, 0.0424)
(92, 0.0424)
(93, 0.0424)
(94, 0.0424)
(95, 0.0425)
(96, 0.0425)
(97, 0.0425)
(98, 0.0425)
(99, 0.0425)
(100, 0.0425)
    };
    % \addlegendentry{$\lambda=0.95$}

    % Plot 2 - lambda
    \nextgroupplot[
        ylabel={$R(n,\lambda)$},
        ylabel style={at={(axis description cs:-0.15,0.5)}},
        xlabel={$n$},
        xlabel style={
            at={(current axis.right of origin)},
            anchor=west
        },
        % Small work around, the limits are quirky
        ymin=0, ymax=1.2,
        ytick={0.001, 0.25, 0.5, 0.75, 1},
        yticklabels={0, 0.25, 0.5, 0.75, 1},
        xmin=0, xmax=100,
        xtick={25, 50, 75, 99.995},
        xticklabels={25, 50, 75, 100}
    ]
    \addplot[
        color=cbblue,
        solid,
    ] coordinates {
(1, 0.3984)
(2, 0.5434)
(3, 0.8075)
(4, 0.7264)
(5, 0.7465)
(6, 0.876)
(7, 0.8306)
(8, 0.8226)
(9, 0.8986)
(10, 0.8835)
(11, 0.8635)
(12, 0.9106)
(13, 0.9165)
(14, 0.8898)
(15, 0.9189)
(16, 0.9396)
(17, 0.9087)
(18, 0.9254)
(19, 0.9572)
(20, 0.9234)
(21, 0.9311)
(22, 0.9713)
(23, 0.9354)
(24, 0.9363)
(25, 0.9717)
(26, 0.9456)
(27, 0.9411)
(28, 0.9678)
(29, 0.9545)
(30, 0.9457)
(31, 0.9653)
(32, 0.9626)
(33, 0.9502)
(34, 0.9639)
(35, 0.9699)
(36, 0.9546)
(37, 0.9634)
(38, 0.9767)
(39, 0.9588)
(40, 0.9634)
(41, 0.9831)
(42, 0.963)
(43, 0.964)
(44, 0.9851)
(45, 0.9672)
(46, 0.9649)
(47, 0.9817)
(48, 0.9713)
(49, 0.9663)
(50, 0.9791)
(51, 0.9754)
(52, 0.9679)
(53, 0.9773)
(54, 0.9795)
(55, 0.9697)
(56, 0.9761)
(57, 0.9835)
(58, 0.9718)
(59, 0.9753)
(60, 0.9876)
(61, 0.974)
(62, 0.9751)
(63, 0.9904)
(64, 0.9764)
(65, 0.9752)
(66, 0.9876)
(67, 0.9789)
(68, 0.9756)
(69, 0.9854)
(70, 0.9815)
(71, 0.9763)
(72, 0.9837)
(73, 0.9842)
(74, 0.9773)
(75, 0.9824)
(76, 0.987)
(77, 0.9785)
(78, 0.9816)
(79, 0.9899)
(80, 0.9798)
(81, 0.9811)
(82, 0.9929)
(83, 0.9814)
(84, 0.9809)
(85, 0.991)
(86, 0.9831)
(87, 0.9809)
(88, 0.989)
(89, 0.985)
(90, 0.9812)
(91, 0.9875)
(92, 0.987)
(93, 0.9818)
(94, 0.9863)
(95, 0.9891)
(96, 0.9826)
(97, 0.9854)
(98, 0.9914)
(99, 0.9835)
(100, 0.9848)
    };
    \addlegendentry{$\lambda=0.25$}
    
    \addplot[
        color=black,
        solid,
    ] coordinates {
(1, 0.4444)
(2, 0.5333)
(3, 0.6857)
(4, 0.8635)
(5, 0.7792)
(6, 0.7752)
(7, 0.8441)
(8, 0.9258)
(9, 0.8604)
(10, 0.8528)
(11, 0.8973)
(12, 0.9486)
(13, 0.8979)
(14, 0.8907)
(15, 0.9235)
(16, 0.9606)
(17, 0.9195)
(18, 0.9131)
(19, 0.9391)
(20, 0.968)
(21, 0.9335)
(22, 0.9279)
(23, 0.9495)
(24, 0.973)
(25, 0.9434)
(26, 0.9384)
(27, 0.9568)
(28, 0.9767)
(29, 0.9508)
(30, 0.9463)
(31, 0.9623)
(32, 0.9795)
(33, 0.9564)
(34, 0.9523)
(35, 0.9665)
(36, 0.9817)
(37, 0.9609)
(38, 0.9572)
(39, 0.9699)
(40, 0.9835)
(41, 0.9645)
(42, 0.9611)
(43, 0.9727)
(44, 0.9849)
(45, 0.9676)
(46, 0.9644)
(47, 0.975)
(48, 0.9861)
(49, 0.9701)
(50, 0.9672)
(51, 0.9769)
(52, 0.9872)
(53, 0.9723)
(54, 0.9695)
(55, 0.9786)
(56, 0.9881)
(57, 0.9742)
(58, 0.9716)
(59, 0.98)
(60, 0.9888)
(61, 0.9758)
(62, 0.9734)
(63, 0.9813)
(64, 0.9895)
(65, 0.9773)
(66, 0.975)
(67, 0.9824)
(68, 0.9901)
(69, 0.9785)
(70, 0.9764)
(71, 0.9834)
(72, 0.9907)
(73, 0.9797)
(74, 0.9776)
(75, 0.9843)
(76, 0.9911)
(77, 0.9807)
(78, 0.9787)
(79, 0.9851)
(80, 0.9916)
(81, 0.9816)
(82, 0.9798)
(83, 0.9858)
(84, 0.992)
(85, 0.9825)
(86, 0.9807)
(87, 0.9864)
(88, 0.9923)
(89, 0.9833)
(90, 0.9815)
(91, 0.987)
(92, 0.9927)
(93, 0.984)
(94, 0.9823)
(95, 0.9876)
(96, 0.993)
(97, 0.9846)
(98, 0.983)
(99, 0.9881)
(100, 0.9932)
    };
    \addlegendentry{$\lambda=0.50$}

    \addplot[
        color=cbblue,
        dashed,
    ] coordinates {
(1, 0.5459)
(2, 0.5994)
(3, 0.6273)
(4, 0.6922)
(5, 0.8033)
(6, 0.9587)
(7, 0.9031)
(8, 0.8519)
(9, 0.8343)
(10, 0.848)
(11, 0.8908)
(12, 0.9604)
(13, 0.961)
(14, 0.9176)
(15, 0.8969)
(16, 0.8976)
(17, 0.9185)
(18, 0.9582)
(19, 0.9879)
(20, 0.9497)
(21, 0.9284)
(22, 0.9231)
(23, 0.9332)
(24, 0.9578)
(25, 0.9964)
(26, 0.9697)
(27, 0.9482)
(28, 0.9395)
(29, 0.9432)
(30, 0.9587)
(31, 0.9856)
(32, 0.984)
(33, 0.9624)
(34, 0.9515)
(35, 0.9509)
(36, 0.9603)
(37, 0.9794)
(38, 0.9951)
(39, 0.9735)
(40, 0.961)
(41, 0.9574)
(42, 0.9624)
(43, 0.9759)
(44, 0.9975)
(45, 0.9826)
(46, 0.969)
(47, 0.9632)
(48, 0.9649)
(49, 0.9741)
(50, 0.9905)
(51, 0.9906)
(52, 0.976)
(53, 0.9685)
(54, 0.9677)
(55, 0.9735)
(56, 0.9858)
(57, 0.9976)
(58, 0.9824)
(59, 0.9734)
(60, 0.9706)
(61, 0.9737)
(62, 0.9828)
(63, 0.9976)
(64, 0.9882)
(65, 0.9781)
(66, 0.9736)
(67, 0.9746)
(68, 0.9809)
(69, 0.9926)
(70, 0.9937)
(71, 0.9826)
(72, 0.9767)
(73, 0.9759)
(74, 0.98)
(75, 0.9889)
(76, 0.9989)
(77, 0.987)
(78, 0.9799)
(79, 0.9775)
(80, 0.9797)
(81, 0.9864)
(82, 0.9976)
(83, 0.9913)
(84, 0.9832)
(85, 0.9795)
(86, 0.98)
(87, 0.9847)
(88, 0.9937)
(89, 0.9955)
(90, 0.9865)
(91, 0.9816)
(92, 0.9807)
(93, 0.9838)
(94, 0.9908)
(95, 0.9997)
(96, 0.9899)
(97, 0.984)
(98, 0.9818)
(99, 0.9834)
(100, 0.9887)
    };
    \addlegendentry{$\lambda=0.75$}
    
    \addplot[
        color=black,
        dashed,
    ] coordinates {
(1, 0.7896)
(2, 0.9483)
(3, 0.9732)
(4, 0.957)
(5, 0.928)
(6, 0.8968)
(7, 0.8682)
(8, 0.844)
(9, 0.8252)
(10, 0.8122)
(11, 0.8049)
(12, 0.8032)
(13, 0.807)
(14, 0.8161)
(15, 0.8303)
(16, 0.8493)
(17, 0.873)
(18, 0.9011)
(19, 0.9335)
(20, 0.9699)
(21, 1.0103)
(22, 1.0544)
(23, 1.1021)
(24, 1.0752)
(25, 1.0464)
(26, 1.0208)
(27, 0.9982)
(28, 0.9785)
(29, 0.9617)
(30, 0.9478)
(31, 0.9365)
(32, 0.9279)
(33, 0.9218)
(34, 0.9182)
(35, 0.9171)
(36, 0.9183)
(37, 0.9218)
(38, 0.9276)
(39, 0.9356)
(40, 0.9457)
(41, 0.9579)
(42, 0.9721)
(43, 0.9884)
(44, 1.0066)
(45, 1.0267)
(46, 1.0487)
(47, 1.0483)
(48, 1.0309)
(49, 1.0152)
(50, 1.0012)
(51, 0.989)
(52, 0.9784)
(53, 0.9694)
(54, 0.9621)
(55, 0.9563)
(56, 0.9521)
(57, 0.9495)
(58, 0.9483)
(59, 0.9486)
(60, 0.9504)
(61, 0.9536)
(62, 0.9582)
(63, 0.9641)
(64, 0.9715)
(65, 0.9802)
(66, 0.9902)
(67, 1.0015)
(68, 1.0141)
(69, 1.0279)
(70, 1.0383)
(71, 1.0256)
(72, 1.0141)
(73, 1.0038)
(74, 0.9946)
(75, 0.9867)
(76, 0.9799)
(77, 0.9742)
(78, 0.9697)
(79, 0.9663)
(80, 0.964)
(81, 0.9627)
(82, 0.9626)
(83, 0.9635)
(84, 0.9654)
(85, 0.9684)
(86, 0.9724)
(87, 0.9774)
(88, 0.9834)
(89, 0.9904)
(90, 0.9983)
(91, 1.0073)
(92, 1.0171)
(93, 1.028)
(94, 1.0232)
(95, 1.014)
(96, 1.0056)
(97, 0.9982)
(98, 0.9917)
(99, 0.9861)
(100, 0.9814)
    };
    \addlegendentry{$\lambda=0.95$}

    \end{groupplot}
    \end{tikzpicture}
    \caption{The performance of the optimal order statistic policy (left). Agreement with the first-order convergence constant (right).}
    \label{fig:performance}
\end{figure}
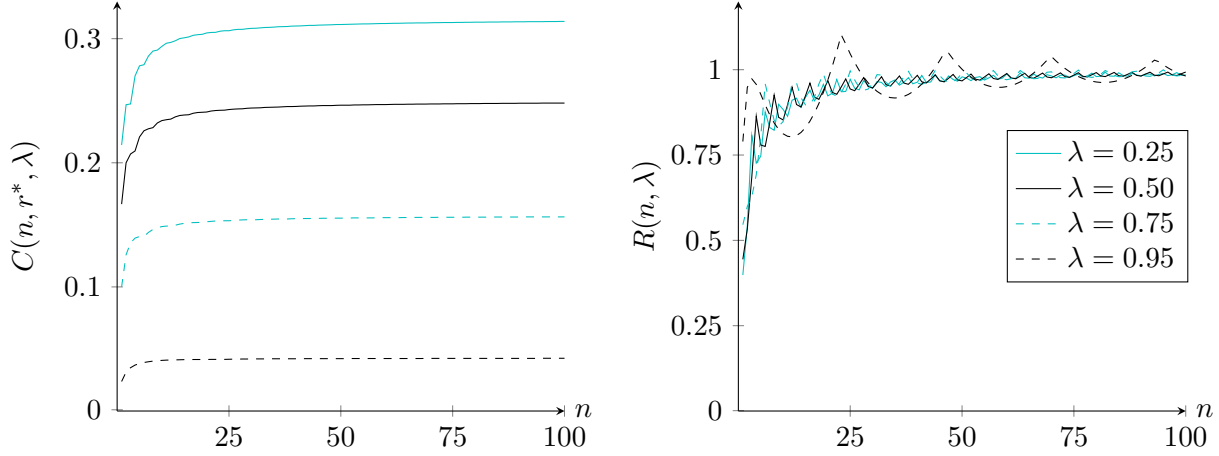

\section{Proof of Theorem \ref{thm:general-pricing-impossibility}}
\label{appendix:proof-general-pricing-impossibility}
Proving Theorem \ref{thm:general-pricing-impossibility} entails showing that no (possibly randomised) sample-based pricing strategy can get closer than $O(1/n)$ to the full-information setting. To prove this we use the same principle as Huang et al. \cite{huang2018making}: 
reduce the problem to a classification problem and use known impossibility results.
Specifically, our results are an application of Le Cam’s two-point method \cite{tsybakov2008lower}.
The proofs of ancillary Lemmas \ref{lemma:ell-quad-bound} and \ref{lemma:kl-div-bound} are deferred to the end of the section.

\begin{proof}[Proof of Theorem \ref{thm:general-pricing-impossibility}]
    
To start, fix $\lambda\in(0,1)$ and let $F_\alpha(x)=1-\left(1+\frac{\lambda}{1-\lambda}\frac{x}{\alpha}\right)^{-1/\lambda}$. Note that $F_\alpha\in \mathcal{P}_\lambda$ and that $\{F_\alpha:\alpha>0\}\subset \mathcal{P}_\lambda$, so
%\pk{I might be messing up the sup's and inf's, but do you not get $$(1-\lambda)^{1/\lambda}-\sup_{\kappa\in \mathcal{K}(\Rplus^n,\Rplus)}\sup_{\alpha>0}\left((1-\lambda)^{1/\lambda}-T(\kappa,F_\alpha)\right),$$ if you rewrite below, because the adversary want to make the gap with $(1-\lambda)^{-1/\lambda}$ as large as possible? }\dn{typo it is inf-sup yes.} \pk{Ah,yes, typo on my end as well. Thanks}
$$
Q_{n,\lambda}\le\sup_{\kappa\in \mathcal{K}(\Rplus^n,\Rplus)}\inf_{F\in\{F_\alpha:\alpha>0\}}T(\kappa,F)=(1-\lambda)^{1/\lambda}-\inf_{\kappa\in \mathcal{K}(\Rplus^n,\Rplus)}\sup_{\alpha>0}\left((1-\lambda)^{1/\lambda}-T(\kappa,F_\alpha)\right).
$$
Using that the mean of $V \sim F_{\alpha}$ is $\alpha$, which can be checked by computing $\E[V] = \int_0^\infty (1-F_\alpha(x))\mathrm{d}x$, we find that
$$
(1-\lambda)^{1/\lambda}-T(\kappa,F_\alpha) = (1-\lambda)^{1/\lambda} - \frac{\E[\Rev(F,Y_{V_1,\dots,V_n})]}{\E[V]} =\E[\ell_\alpha(Y_{\boldsymbol{V}})],
$$
where 
$$\ell_\alpha(t)=(1-\lambda)^{1/\lambda}-\frac{1}{\alpha}\Rev(F_\alpha,t)=(1-\lambda)^{1/\lambda}-\frac{t}{\alpha}\left(1+\frac{\lambda}{1-\lambda}\frac{t}{\alpha}\right)^{-1/\lambda}.$$
and $Y_{\boldsymbol{V}} = Y_{V_1,\dots,V_n}$ a (possibly randomised) pricing rule. Define $\ell(t)=\ell_\alpha(\alpha t)$

In what comes next, we will sometimes use the following result on the loss function $\ell$, that yields $\ell_\alpha(\alpha+t)\ge c_0\cdot(t/\alpha)^2$ for $t\in[-\alpha,\alpha/4]$ with $c_0$ a constant that depends only on $\lambda$. 
% \pk{Not completely sure if I understand the application of the Lemma, because if I would do the proof with $\alpha$ instead of $1$, do I not get a constant depending on $\alpha$? I might be missing something simple, but I do not directly see how substitutions get you from $\ell$ to $\ell_\alpha$ for general $\alpha$.}
\begin{lemma}
    \label{lemma:ell-quad-bound}
    Let $\lambda\in(0,1)$ and define
    $$\ell(t):=(1-\lambda)^{1/\lambda}-t\left(1+\frac{\lambda}{1-\lambda}t\right)^{-1/\lambda}.$$
    Then for all $t\in[-1,1/4]$, we have that $\ell(1+t)\ge \frac{1}{2}\ell''(5/4)t^2$.
\end{lemma}

Let $\alpha_0=1$, $\alpha_n=1+t_n$, $t_n=h/\sqrt{n}$, and $h=1/2$. Let $m_n=(\alpha_0+\alpha_n)/2$ and define $\phi_n=\indicator{Y_{\boldsymbol{V}}\ge m_n}$. We use the so-called test $\phi_n$ to derive lower bounds on $\E[\ell_{a_j}(Y_{\boldsymbol{V}})]$ for $j = 0,n$. There are two cases to consider.
First, consider the event $\{\phi_n=1\}$. In that case, 
$$Y_{\boldsymbol{V}}\indicator{\phi_n = 1}\ge m_n\indicator{\phi_n = 1}=(\alpha_0+t_n/2)\indicator{\phi_n = 1}.$$ 
Because the loss function $\ell_{\alpha_0}$ is increasing after $\alpha_0$, we have that 
$$\ell_{\alpha_0}(Y_{\boldsymbol{V}})\indicator{\phi_n = 1}\ge\ell_{\alpha_0}(\alpha_0+t_n/2)\indicator{\phi_n = 1}.$$ 
Note that $t_n/2=1/(4\sqrt{n})\le 1/4= \alpha_0/4$, so $\ell_{\alpha_0}(\alpha_0+t_n/2)\ge \frac{1}{4}c_0(t_n/\alpha_0)^2=\frac{1}{4}c_0t_n^2$. This yields
\begin{align*}
    \E[\ell_{\alpha_0}(Y_{\boldsymbol{V}})]&\ge \E[\ell_{\alpha_0}(Y_{\boldsymbol{V}})\indicator{\phi_n=1}]
    \\
    &\ge \E\left[\frac{c_0}{4}t_n^2\indicator{\phi_n=1}\right]=\frac{c_0}{4}t_n^2\Prob_{\alpha_0}(\phi_n=1).
\end{align*}
Second, consider the event $\{\phi_n=0\}$. In that case, $Y_{\boldsymbol{V}}<m_n=\alpha_n-t_n/2$. Because the loss function $\ell_{\alpha_n}$ is decreasing before $\alpha_n$, we have that 
$$
\ell_{\alpha_n}(Y_{\boldsymbol{V}})\indicator{\phi_n = 0}\ge \ell_{\alpha_n}(\alpha_n-t_n/2)\indicator{\phi_n = 0}\ge\frac{1}{4}c_0\frac{t_n^2}{a_n^2}\indicator{\phi_n = 0}\ge \frac{1}{9}c_0t_n^2\indicator{\phi_n = 0}.
$$
Using $\alpha_n\le3/2$ for the last inequality. We then analogously find
$$
\E[\ell_{\alpha_n}(Y_{\boldsymbol{V}})]\ge \frac{1}{9}c_0t_n^2\Prob_{\alpha_n}(\phi_n=0).
$$
With these steps we can do the main step.
\begin{align*}
    \inf_{\kappa\in \mathcal{K}(\Rplus^n,\Rplus)}\sup_{\alpha>0}\E[\ell_{\alpha}(Y_{\boldsymbol{V}})]&\ge \inf_{\kappa\in \mathcal{K}(\Rplus^n,\Rplus)}\max\{\E[\ell_{\alpha_0}(Y_{\boldsymbol{V}})],\E[\ell_{\alpha_n}(Y_{\boldsymbol{V}})]\}\\
    &\ge \frac{1}{2}\inf_{\kappa\in \mathcal{K}(\Rplus^n,\Rplus)}(\E[\ell_{\alpha_0}(Y_{\boldsymbol{V}})]+\E[\ell_{\alpha_n}(Y_{\boldsymbol{V}})])\\
    &\ge \frac{1}{2}\inf_{\kappa\in \mathcal{K}(\Rplus^n,\Rplus)}\left(\frac{c_0}{4}t_n^2\Prob_{\alpha_0}(\phi_n=1)+\frac{c_0}{9}t_n^2\Prob_{\alpha_n}(\phi_n=0)\right)\\
    &\ge \frac{c_0}{18}t_n^2\inf_{\kappa\in \mathcal{K}(\Rplus^n,\Rplus)}\left(\Prob_{\alpha_0}(\phi_n=1)+\Prob_{\alpha_n}(\phi_n=0)\right)\\
    &= \frac{c_0}{18}t_n^2\inf_{\kappa\in \mathcal{K}(\Rplus^n,\Rplus)}\left(1+\Prob_{\alpha_0}(\phi_n=1)-\Prob_{\alpha_n}(\phi_n=1)\right)\\
    &= \frac{c_0}{18}t_n^2\left(1-\sup_{\kappa\in \mathcal{K}(\Rplus^n,\Rplus)}\left(\Prob_{\alpha_n}(\phi_n=1)-\Prob_{\alpha_0}(\phi_n=1)\right)\right)\\
    &\ge  \frac{c_0}{18}t_n^2\left(1-\vert\vert P_{\alpha_n}-P_{\alpha_0} \vert\vert_{\mathrm{TV}}\right).
\end{align*}
The last step follows from the fact that the total variation distance provides an upper bound on the difference between the probabilities of any given event under two distributions. Using Pinsker's inequality and the additivity of the Kullback-Leibler divergence yields
$$
\vert\vert P_{\alpha_n}-P_{\alpha_0}\vert\vert_{\mathrm{TV}}\le\sqrt{\frac{1}{2}D_{\mathrm{KL}}(P_{\alpha_0}||P_{\alpha_n})}=\sqrt{\frac{1}{2}nD_{\mathrm{KL}}(F_{\alpha_0}||F_{\alpha_n})}.
$$
We next use the following lemma to upper bound the Kullback-Leibler divergence.

\begin{lemma}
    \label{lemma:kl-div-bound}
    Let $\lambda\in(0,1)$, $F_\alpha(x)=1-\left(1+\frac{\lambda}{1-\lambda}\frac{x}{\alpha}\right)^{-1/\lambda}$ and let $\epsilon \in (0, 1)$, then
    $$
    D_{\mathrm{KL}}(F_{1}||F_{1+\epsilon})\le \frac{\epsilon^2}{1-\epsilon}.
    $$
\end{lemma}
Applying Lemma \ref{lemma:kl-div-bound} with $\epsilon = t_n$, we obtain
$$
\vert\vert P_{\alpha_n}-P_{\alpha_0}\vert\vert_{\mathrm{TV}}\le\sqrt{\frac{n}{2}\frac{t_n^2}{1-t_n}}=\sqrt{\frac{h^2}{2\left(1-h/\sqrt{n}\right)}}\le \sqrt{\frac{h^2}{2(1-h)}}=\frac{1}{2}.
$$
Using $h=1/2$ in the last step. With this bound on the total variation distance we find that
$$
    \inf_{\kappa\in \mathcal{K}(\Rplus^n,\Rplus)}\sup_{\alpha>0}\E[\ell_{\alpha}(Y_{\boldsymbol{V}})]\ge  \frac{c_0}{36}t_n^2=\frac{c_0}{144}\frac{1}{n}.
$$
This shows that
\begin{align*}
    Q_{n,\lambda}&\le (1-\lambda)^{1/\lambda}-\frac{c_0}{144}\frac{1}{n}=(1-\lambda)^{1/\lambda}-\frac{\ell''(5/4)}{288}\frac{1}{n}\\
    &=(1-\lambda)^{1/\lambda}-\frac{1}{384\left(1-\lambda\right)}\left(1+\frac{\lambda}{1-\lambda}\frac{5}{4}\right)^{-\frac{1}{\lambda}-2}\frac{1}{n}.  \qedhere
\end{align*}
\end{proof}

We conclude this section with the proofs of Lemmas \ref{lemma:ell-quad-bound} and \ref{lemma:kl-div-bound}.

\begin{proof}[Proof of \ref{lemma:ell-quad-bound}]
    The derivatives of $\ell$ are
    \begin{align*}
        \ell'(t)&=(1-t)\left(1+\frac{\lambda}{1-\lambda}t\right)^{-\frac{1}{\lambda}-1},\\
        \ell''(t)&=\frac{2-t}{1-\lambda}\left(1+\frac{\lambda}{1-\lambda}t\right)^{-\frac{1}{\lambda}-2},\\
        \ell'''(t)&=\frac{\left(1+\lambda\right)\left(t-3\right)}{\left(1-\lambda\right)^{2}}\left(1+\frac{\lambda}{1-\lambda}t\right)^{-\frac{1}{\lambda}-3}.
    \end{align*}
    Note that $\ell(1)=0$ and that $\ell'(1)=0$. By Taylor's theorem, for each $t$ there exists a
    $\xi_t\in(1,1+t)$ (or $\xi_t\in(1+t,1)$ if $t<0$) such that
    $$
    \ell(1+t)=\ell(1)+\ell'(1)t+\frac{1}{2}\ell''(\xi_t)t^2=\frac{1}{2}\ell''(\xi_t)t^2.
    $$
    Now, $\ell'''(s)<0$ for $s\in(0,3)$, so $\ell''$ is decreasing on that interval. Consequently, for $t\in[-1,1/4]$ it holds that $\ell''(\xi_t)\ge \ell''(1+t)\ge \ell''(5/4)$. So for $t\in[-1,1/4]$,
    \begin{equation*}
        \ell(1+t)\ge \frac{1}{2}\ell''(5/4)t^2.\qedhere
    \end{equation*}
\end{proof}

\medskip

\begin{proof}[Proof of Lemma \ref{lemma:kl-div-bound}]
    Let $A(x)=\frac{1-\lambda}{\lambda}\left(x^{-\lambda}-1\right)$
denote the inverse of $1-F_{1}(x)$. A direct computation informs us that
\begin{align*}
    D_{\mathrm{KL}}(F_{1}||F_{1+\epsilon})&:=\int_0^\infty f_{1}(t)\log\frac{f_{1}(t)}{f_{1+\epsilon}(t)}\dd t
    =\int_0^1 \log\frac{f_{1}(A(x))}{f_{1+\epsilon}(A(x))}\dd x\\
    &=\int_0^1 \log
    \frac{\frac{1}{1-\lambda}\left(1+\frac{\lambda}{1-\lambda}A(x)\right)^{-1/\lambda-1}}{\frac{1}{(1-\lambda)(1+\epsilon)}\left(1+\frac{\lambda}{1-\lambda}\frac{A(x)}{1+\epsilon}\right)^{-1/\lambda-1}}\dd x\\
    &=\int_0^1 \log
    \frac{x^{1+\lambda}}{\frac{1}{1+\epsilon}\left(1+\frac{x^{-\lambda}-1}{1+\epsilon}\right)^{-1/\lambda-1}}\dd x\\
    &=(1+\lambda)\int_0^1\log x\;\dd x-\int_0^1 \log
    \frac{1}{1+\epsilon}\left(1+\frac{x^{-\lambda}-1}{1+\epsilon}\right)^{-1/\lambda-1}\dd x\\
    &=\log(1+\epsilon)-(1+\lambda)+\frac{1+\lambda}{\lambda}\int_0^1 \log
    \left(1+\frac{x^{-\lambda}-1}{1+\epsilon}\right)\dd x\\
    &=-\frac{1}{\lambda}\log(1+\epsilon)-(1+\lambda)+\frac{1+\lambda}{\lambda}\int_0^1 \log
    \left(\epsilon+x^{-\lambda}\right)\dd x\\
    &=\sum_{k=2}^{\infty}\frac{\left(-1\right)^{k}\left(k-1\right)}{k\left(\lambda k+1\right)}\epsilon^{k}\le \sum_{k=2}^{\infty}\epsilon^{k}=\frac{\epsilon^2}{1-\epsilon}.\qedhere
\end{align*}
\end{proof}

\end{document}